\documentclass[a4paper,USenglish,cleveref, autoref, thm-restate, numberwithinsect]{lipics-v2021}

\hideLIPIcs

\title{A Graph Width Perspective on Partially Ordered Hamiltonian Paths and Cycles I}
\subtitle{Treewidth, Pathwidth, and Grid Graphs} %TODO Please add

\titlerunning{A Graph Width Perspective on Partially Ordered Hamiltonian Paths and Cycles I} 

\author{Jesse Beisegel}{Institute of Mathematics, Brandenburg University of Technology, Cottbus, Germany}{jesse.beisegel@b-tu.de}{https://orcid.org/0000-0002-8760-0169}{}
\author{Katharina Klost}{Institute of Computer Science, Freie Universität Berlin, Germany}{katharina.klost@fu-berlin.de}{https://orcid.org/0000-0002-9884-3297}{}
\author{Kristin Knorr}{Institute of Computer Science, Freie Universität Berlin, Germany}{kristin.knorr@fu-berlin.de}{https://orcid.org/0000-0003-4239-424X}{}
\author{Fabienne Ratajczak}{Institute of Mathematics, Brandenburg University of Technology, Cottbus, Germany}{fabienne.ratajczak@b-tu.de}{https://orcid.org/0000-0002-5823-1771}{The research was funded by the \emph{Federal Ministry for Digital and Transport Germany} through the research project \emph{MoVeToLausitz} (grant number 19FS2032C).}
\author{Robert Scheffler}{Institute of Mathematics, Brandenburg University of Technology, Cottbus, Germany}{robert.scheffler@b-tu.de}{https://orcid.org/0000-0001-6007-4202}{}

\authorrunning{J. Beisegel, K. Klost, K. Knorr, F. Ratacjzak, and R. Scheffler} 

\Copyright{Jesse Beisegel, Katharina Klost, Kristin Knorr, Fabienne Ratacjzak, and Robert Scheffler}


\relatedversion{} 

\acknowledgements{We want to thank Philipp Wolf Schleicher for his helpful suggestions.}

\nolinenumbers %uncomment to disable line numbering

\hideLIPIcs

\theoremstyle{definition}
\newtheorem{problem}{Problem}

\crefname{observation}{Observation}{Observations}
\Crefname{observation}{Observation}{Observations}


\newcommand{\param}[1]{{#1}}

\renewcommand{\P}{\ensuremath{\mathsf{P}}}
\newcommand{\NP}{\ensuremath{\mathsf{NP}}}
\newcommand{\XP}{\ensuremath{\mathsf{XP}}{}}
\newcommand{\FPT}{\ensuremath{\mathsf{FPT}}}
\newcommand{\W}{\ensuremath{\mathsf{W[1]}}}
\newcommand{\MSO}{\ensuremath{\mathsf{MSO}}}
\newcommand{\cP}{\ensuremath{\mathcal{P}}}
\newcommand{\A}{\ensuremath{\mathcal{A}}}
\renewcommand{\O}{\ensuremath{\mathcal{O}}}
\newcommand{\G}{\ensuremath{\mathcal{G}}}

\newcommand{\C}{\ensuremath{\mathcal{C}}}
\newcommand{\B}{\ensuremath{\mathcal{B}}}
\newcommand{\N}{\ensuremath{\mathbb{N}}}

\newcommand{\pef}[1]{(P\ref{#1})}

\usepackage{enumerate}
\usepackage{xcolor}

\usepackage{tikz}
\usetikzlibrary{decorations.pathreplacing,arrows,math,shapes}
\usetikzlibrary{decorations.pathmorphing}
\tikzstyle{vertex}=[draw, circle, fill=black, inner sep=1.5pt]
\definecolor{variables}{RGB}{0,102,204}
\definecolor{clauses}{RGB}{204,0,0}
\definecolor{backbone}{RGB}{255,153,0}
\definecolor{variables2}{RGB}{153,204,0}
\definecolor{clauses2}{RGB}{0,100,0}
\definecolor{hard}{RGB}{204,0,0}
\definecolor{easy}{RGB}{142,196,222}

\begin{document}

\maketitle

\begin{abstract}
We consider the problem of finding a Hamiltonian path or a Hamiltonian cycle with precedence constraints in the form of a partial order on the vertex set. We show that the path problem is \NP-complete for graphs of pathwidth~4 while the cycle problem is \NP-complete on graphs of pathwidth~5. We complement these results by giving polynomial-time algorithms for graphs of pathwidth~3 and treewidth~2 for Hamiltonian paths as well as pathwidth~4 and treewidth~3 for Hamiltonian cycles. Furthermore, we study the complexity of the path and cycle problems on rectangular grid graphs of bounded height. For these, we show that the path and cycle problems are \NP-complete when the height of the grid is greater or equal to~7 and~9, respectively. In the variant where we look for minimum edge-weighted Hamiltonian paths and cycles, the problems are \NP-hard for heights~5 and~6, respectively. \end{abstract}

\section{Introduction}

The problems \textsc{Hamiltonian Cycle} and \textsc{Hamiltonian Path} as well as their weighted versions \textsc{Traveling Salesperson Problem} (TSP) and \textsc{Path Traveling Salesperson Problem} (Path TSP) belong to the most famous \NP-hard graph problems and have a wide range of applications. In some of these applications it is necessary to fulfill certain precedence constraints, i.e., one vertex has to be visited before another. For a long time, the research on Hamiltonian paths and cycles with such precedence constraints has focused on the case where the given graph is a complete graph with certain weights on the edges (see \cref{sec:related-work} for references). In contrast, the problems without constraints have been extensively studied also for the unweighted case on a wide range of graph classes and graph width parameters. 
To combine these two research directions, Beisegel et al.~\cite{beisegel2024computing} introduced the \textsc{Partially Ordered Hamiltonian Path Problem (POHPP)} where we are given a graph together with a partial order on its vertices and search for a Hamiltonian path that is a linear extension of the partial order.

Here, we will continue this research by considering the problem on some famous graph classes and graph width parameters, namely proper interval graphs and grid graphs as well as graphs of bounded \param{treewidth}, \param{pathwidth}, and \param{bandwidth}. In contrast to Beisegel et al.~\cite{beisegel2024computing}, we do not only focus on the Hamiltonian path case, but also consider the cycle case, i.e., the \textsc{Partially Ordered Hamiltonian Cycle Problem (POHCP)}. In their conclusion, Beisegel et al.~\cite{beisegel2024computing} suggested to use (partial) cyclic orders to define this problem. Such cyclic orders consist of triplets that force one vertex to be between the other two in the cycle. However, the authors also mentioned that such a variant would imply certain difficulties. In contrast to partial orders, cyclic orders are not necessarily extendable, i.e., there does not exist a \emph{total} cyclic ordering of the vertices that fulfill all the constraints. Even worse, the problem of deciding whether such an extension exists is \NP-complete~\cite{galil1978cyclic}. To circumvent this issue for TSP, the authors of \cite{ahmed2001travelling} and \cite{bianco1994exact} assume that the cycle must begin and end in a certain vertex (the headquarter) and the vertices in between form a linear extension of a given partial order. We adapt this approach slightly for the Hamiltonian cycle problem by searching for a Hamiltonian path that forms a linear extension of a given partial order and whose endpoints are adjacent.

\subsection{Related Work}\label{sec:related-work}

\subparagraph{Hamiltonian Paths and Cycles with Precedence}

Both the cycle and the path variant of the TSP have been considered together with precedence constraints. The cycle variant is known as \textsc{Traveling Salesman Problem with Precedence Constraints} (TSP-PC) and has been studied, e.g., in \cite{ahmed2001travelling,bianco1994exact}. The path variant, known as the \textsc{Sequential Ordering Problem} (SOP) or the \textsc{Minimum Setup Scheduling Problem}, has been studied, e.g., in~\cite{ascheuer1993cutting,colbourn1985minimizing,escudero1988inexact,escudero1988implementation}. 

Of course, all these problems are \NP-hard and the research focus has mainly been on heuristics and integer-programming approaches. Furthermore, these problems are defined over complete graphs with an additional cost function. The unweighted variants \textsc{Hamiltonian Path} and \textsc{Hamiltonian Cycle} with precedence constraints for non-complete graphs have not received the same level of attention for a long time. Results have been only given for the very restricted variants where one or both endpoints of the Hamiltonian path are fixed. For these problems, polynomial-time algorithms have been presented for several graph classes including (proper) interval graphs~\cite{asdre2010fixed,asdre2010polynomial,li2017linear,mertzios2010optimal}, distance-hereditary graphs~\cite{hsieh2004efficient,yeh1998path}, and rectangular grid graphs~\cite{itai1982hamilton}.

To overcome this lack of research, Beisegel et al.~\cite{beisegel2024computing} introduced the \textsc{Partially Ordered Hamiltonian Path Problem} (POHPP), where we are given a graph together with a partial order $\pi$ and search for a Hamiltonian path that is a linear extension of $\pi$, as well as the edge-weighted variant \textsc{Minimum Partially Ordered Hamiltonian Path Problem} (MinPOHPP). The authors show that POHPP is already \NP-hard for complete bipartite graphs and complete split graphs -- graph classes where \textsc{Hamiltonian Path} is trivial.
They also show that  POHPP is \W-hard when parameterized by the partial order's \param{width}, i.e., the largest number of pairwise incomparable elements, and that the \XP{} algorithm for that parametrization presented in 1985 by Colbourn and Pulleyblank~\cite{colbourn1985minimizing} is asymptotically optimal -- assuming the Exponential Time Hypothesis (ETH). They improve the algorithm to \FPT{} time if the problem is parameterized by the partial order's \param{distance to linear order}. Finally, they give a polynomial-time algorithm for MinPOHPP on outerplanar graphs.

\subparagraph{Hamiltonicity and Treewidth} As \textsc{Hamiltonian Path} and \textsc{Hamiltonian Cycle} belong to the most famous \NP-hard graph problems, their complexity has been studied for a wide range of graph classes and graph width parameters.\footnote{In general, \textsc{Hamiltonian Path} seems to lead a shadowy existence since many positive and negative algorithmic results are only given for its more popular sibling \textsc{Hamiltonian Cycle}.} In particular, these problems have gained attention for \param{treewidth} and related parameters.

One of the first results were polynomial-time algorithms for \textsc{Hamiltonian Cycle} and TSP on graphs of bounded \param{bandwidth}~\cite{lawler1985traveling,monien1980bounding}, a parameter that is more restrictive than \param{treewidth}. Since both \textsc{Hamiltonian Cycle} and \textsc{Hamiltonian Path} are expressible in $\MSO_2$ logic, Courcelle's theorem~\cite{courcelle1990monadic,courcelle1992monadic} implies \FPT{} algorithms for these problems when parameterized by \param{treewidth}, i.e., there is a computable function $f$ such that the problems can be solved in time $f(k) \cdot n^{\O(1)}$ on graphs with $n$ vertices and \param{treewidth}~$k$. However, the function $f$ given in Courcelle's theorem is far from being practical. Nevertheless, using standard approaches one can give an $k^{\O(k)} \cdot n$ algorithm for \textsc{Hamiltonian Cycle} and \textsc{Hamiltonian Path} on graphs of \param{treewidth} $k$.\footnote{There seems to be no broader description of this algorithm available; it is asked for in Exercise~7.19 of the textbook of Cygan et al.~\cite{cygan2015param}. A short sketch of the algorithm for \textsc{Hamiltonian Cycle} is given by Ziobro and Pilipczuk~\cite{ziobro2019finding}.} For other graph problems such as \textsc{Independent Set} and \textsc{Vertex Cover}, algorithms with running times $2^{\O(k)} \cdot n$ have been known for a long time. Contrarily, it was open for a longer time whether similar algorithms can be also given for non-local problems such as \textsc{Hamiltonian Cycle} and \textsc{Hamiltonian Path}. Such algorithms have been presented in~\cite{bodlaender2015deterministic,cygan2018fast,cygan2022solving}. An experimental study of these approaches can be found in \cite{ziobro2019finding}.

\subparagraph{Hamiltonicity and (Proper) Interval Graphs} Interval graphs, i.e., the intersection graphs of intervals on the real line, have a very restricted structure. This allows for linear-time algorithms for both \textsc{Hamiltonian Path} and \textsc{Hamiltonian Cycle}~\cite{damschke1993paths,keil1985finding}. The subclass of proper interval graphs, where intervals are forbidden to be properly included in each other, has an even simpler structure which implies that such a graph has a Hamiltonian path or cycle if and only if it is connected or 2-connected, respectively~\cite{bertossi1983finding}. This property also holds for the more general \emph{semi-proper interval graphs}~\cite{scheffler2025semi}. Both problems can also be solved in polynomial time on graphs that have $k$ vertices whose deletion constructs a proper interval graph~\cite{golovach2020graph}.

As complete split graphs form a subclass of interval graphs, the results given by Beisegel et al.~\cite{beisegel2024computing} imply that the POHPP is \NP-complete on interval graphs. The subcase where only one endpoint of the Hamiltonian path is fixed, can be solved in linear time~\cite{li2017linear}. Surprisingly, the complexity of the case where both endpoints of the path are fixed is still open on interval graphs. However, for proper interval graphs this problem can be solved in linear time~\cite{asdre2010polynomial,mertzios2010optimal}.

\subparagraph{Hamiltonicity and Grids} For grid graphs, i.e., induced subgraphs of the infinite grid, \textsc{Hamiltonian Path} and \textsc{Hamiltonian Cycle} are \NP-complete~\cite{itai1982hamilton}. If we restrict ourselves to rectangular grid graphs, then both problems become trivial. Rectangular grid graphs always have a Hamiltonian path. Furthermore, they have a Hamiltonian cycle if and only if they are 2-connected and have an even number of vertices. Itai et al.~\cite{itai1982hamilton} showed that one can also solve \textsc{Hamiltonian Path} on rectangular grid graphs if the end vertices of the path are fixed. Their algorithm was later improved to a linear running time by Chen et al.~\cite{chen2002efficient}. In contrast to these results, the complexity of \textsc{Hamiltonian Path} on general solid grid graphs, i.e., grid graphs without holes, seems to be still open while \textsc{Hamiltonian Cycle} has been shown to be polynomial-time solvable on that class by Umans and Lenhart~\cite{umans1997hamiltonian}. Keshavarz-Kohjerdi and Bagheri considered the Hamiltonian path problem with fixed endpoints on several subclasses of solid grid graphs, namely \emph{$C$-shaped}~\cite{keshavarz-kohjerdi2020linear}, \emph{$L$-shaped}~\cite{keshavarz-kohjerdi2016hamiltonian} as well as \emph{$E$-alphabet} and \emph{$F$-alphabet} grid graphs~\cite{keshavarz-kohjerdi2012hamiltonian}.

\subsection{Our Contribution}

We study the complexity of the POHPP and POHCP for graphs of bounded \param{treewidth}, \param{pathwidth} and \param{bandwidth} (see \cref{fig:results} for a summary). In particular, we show that POHPP is \NP-complete on proper interval graphs of \param{clique number}~5 implying that the problem is \NP-complete on graphs of \param{bandwidth}~4. For POHCP, we show the same for proper interval graphs of \param{clique number}~6 and graphs of \param{bandwidth}~5. 

We also show that these hardness results are tight (if $\P \neq \NP$). To this end, we present polynomial-time algorithms for MinPOHPP and MinPOHCP on graphs of \param{pathwidth}~3 and~4, respectively. Additionally, we present algorithms for graphs of \param{treewidth}~2 and~3, respectively, leaving open the case of (Min)POHPP for graphs of \param{treewidth}~3 as well as (Min)POHCP of \param{treewidth}~4.

One might wonder if there are more structured graph classes with bounded \param{bandwidth} or \param{pathwidth} for which the problem becomes easier. For the very structured class of rectangular grid graphs, we are able to show that POHPP is \NP-complete if the height of the grid graph is~7 and MinPOHPP is \NP-hard if the height is~5. When considering the cycle case, then the unweighted problem POHCP is \NP-complete for height~9 and the weighted problem MinPOHCP is \NP-hard for height~6. Note that we can solve MinPOHPP on grids of height~3 and MinPOHCP on grids of height~4 as the height of the grid graph is an upper bound on its \param{pathwidth}.

\begin{figure}[t]
    \centering
    \begin{tikzpicture}[scale=0.8]
    \footnotesize
    \node at (5.5,5.75) {parameter values};
    
    \foreach \x in {1,2,...,10} {
        \draw[lightgray] (\x,-6) -- (\x,4.5);
        \node at (\x,5) {$\x$};
    }
    \node[align=center] at (-1,3.375) {bandwidth / \\ pathwidth};
    
    \draw[fill=easy] (0.5,4) rectangle (3.5,4.25);
    \draw[fill=hard] (3.5,4) rectangle (10.5,4.25);
    \node[align=left] at (12,4.125) {POHPP\phantom{Min}};
    
    \draw[fill=easy] (0.5,3.5) rectangle (3.5,3.75);
    \draw[fill=hard] (3.5,3.5) rectangle (10.5,3.75);
    \node[align=left]  at (12,3.625) {MinPOHPP};
    
    \draw[fill=easy] (0.5,3) rectangle (4.5,3.25);
    \draw[fill=hard] (4.5,3) rectangle (10.5,3.25);
    \node[align=left]  at (12,3.125) {POHCP\phantom{Min}};
    
    \draw[fill=easy] (0.5,2.5) rectangle (4.5,2.75);
    \draw[fill=hard] (4.5,2.5) rectangle (10.5,2.75);
    \node[align=left] at (12,2.625) {MinPOHCP};

    \draw[thick] (-1,2) -- (12,2);

    \begin{scope}[yshift=-2.75cm]
        \node[align=center] at (-1,3.375) {treewidth};
    
        \draw[fill=easy] (0.5,4) rectangle (2.5,4.25);
        \draw[fill=hard] (3.5,4) rectangle (10.5,4.25);
        \node[align=left] at (12,4.125) {POHPP\phantom{Min}};
        
        \draw[fill=easy] (0.5,3.5) rectangle (2.5,3.75);
        \draw[fill=hard] (3.5,3.5) rectangle (10.5,3.75);
        \node[align=left]  at (12,3.625) {MinPOHPP};
        
        \draw[fill=easy] (0.5,3) rectangle (3.5,3.25);
        \draw[fill=hard] (4.5,3) rectangle (10.5,3.25);
        \node[align=left]  at (12,3.125) {POHCP\phantom{Min}};
        
        \draw[fill=easy] (0.5,2.5) rectangle (3.5,2.75);
        \draw[fill=hard] (4.5,2.5) rectangle (10.5,2.75);
        \node[align=left] at (12,2.625) {MinPOHCP};
    
        \draw[thick] (-1,2) -- (12,2);
    \end{scope}

    \begin{scope}[yshift=-5.5cm]
        \node[align=center] at (-1.1,3.375) {height of \\ rectangular \\ grid graph};
    
        \draw[fill=easy] (0.5,4) rectangle (3.5,4.25);
        \draw[fill=hard] (6.5,4) rectangle (10.5,4.25);
        \node[align=left] at (12,4.125) {POHPP\phantom{Min}};
        
        \draw[fill=easy] (0.5,3.5) rectangle (3.5,3.75);
        \draw[fill=hard] (4.5,3.5) rectangle (10.5,3.75);
        \node[align=left]  at (12,3.625) {MinPOHPP};
        
        \draw[fill=easy] (0.5,3) rectangle (4.5,3.25);
        \draw[fill=hard] (8.5,3) rectangle (10.5,3.25);
        \node[align=left]  at (12,3.125) {POHCP\phantom{Min}};
        
        \draw[fill=easy] (0.5,2.5) rectangle (4.5,2.75);
        \draw[fill=hard] (5.5,2.5) rectangle (10.5,2.75);
        \node[align=left] at (12,2.625) {MinPOHCP};
    
        \draw[thick] (-1,2) -- (12,2);
    \end{scope}

    \begin{scope}[yshift=-8.25cm]
        \node[align=center] at (-1.1,3.375) {outerplanarity};
    
        \draw[fill=easy] (0.5,4) rectangle (1.5,4.25);
        \draw[fill=hard] (3.5,4) rectangle (10.5,4.25);
        \node[align=left] at (12,4.125) {POHPP\phantom{Min}};
        
        \draw[fill=easy] (0.5,3.5) rectangle (1.5,3.75);
        \draw[fill=hard] (2.5,3.5) rectangle (10.5,3.75);
        \node[align=left]  at (12,3.625) {MinPOHPP};
        
        \draw[fill=easy] (0.5,3) rectangle (1.5,3.25);
        \draw[fill=hard] (4.5,3) rectangle (10.5,3.25);
        \node[align=left]  at (12,3.125) {POHCP\phantom{Min}};
        
        \draw[fill=easy] (0.5,2.5) rectangle (1.5,2.75);
        \draw[fill=hard] (2.5,2.5) rectangle (10.5,2.75);
        \node[align=left] at (12,2.625) {MinPOHCP};% \beg
    \end{scope}
\end{tikzpicture}
    \caption{Overview of the complexity results and open cases. The light blue bars stand for polynomial-time algorithms, the dark red bars stand for \NP-hardness. The parameter values that are not covered by red or blue bars are open cases.}
    \label{fig:results}
\end{figure}
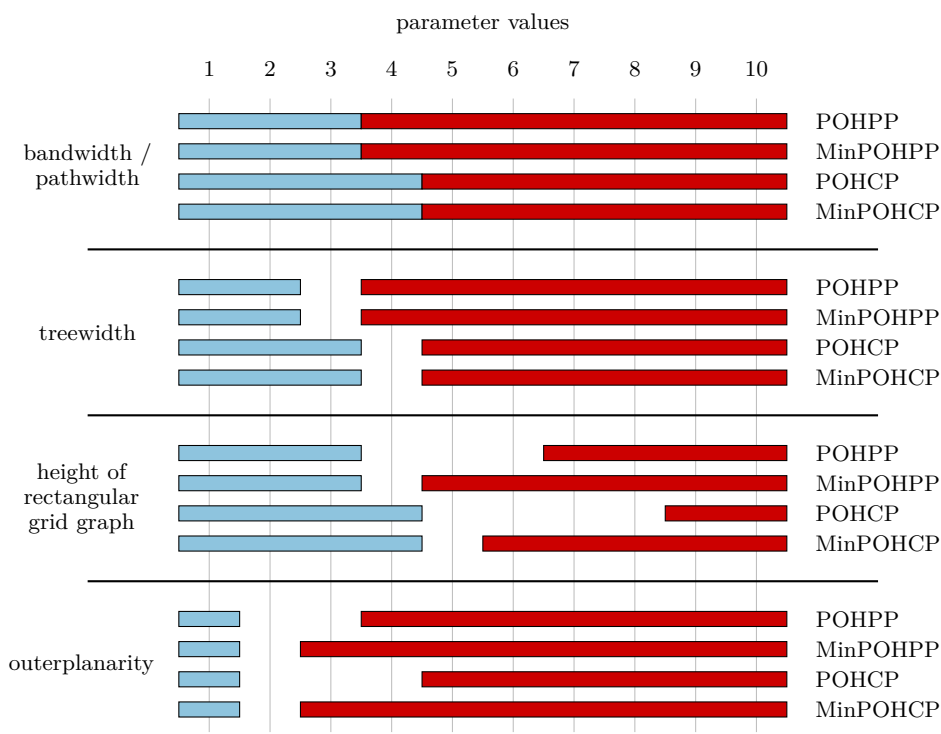

\section{Preliminaries}

\subparagraph{Graphs and Graph Classes}
All the graphs considered here are finite. For the standard notation of graphs we refer to the book of Diestel~\cite{diestel}. A graph $G$ is a \emph{proper interval graph} if it has a vertex ordering $(v_1, \dots, v_n)$ such that for all edges $v_iv_k$ and every $j \in \{i, \dots, k\}$ it holds that $v_iv_j$ and $v_jv_k$ are in $E(G)$. We call such an ordering a \emph{proper interval ordering}. A graph is a \emph{complete split graph} if it can be partitioned into a clique $C$ and an independent set $I$ such that all edges between $C$ and $I$ are present in the graph.

A graph is \emph{planar} if it has a crossing-free embedding in the plane, and together with this embedding it is called a \emph{plane graph}. For a plane graph $G$ we call the regions of $\mathbb{R}^2 \setminus G$ the \emph{faces} of $G$. Every plane graph has exactly one unbounded face which is called the \emph{outer face}. The \emph{outerplanarity} of an embedding is the number of times all vertices of the outer face have to be deleted to create the graph without vertices. The \emph{outerplanarity} of a planar graph is the smallest outerplanarity of one of its embeddings.

A \emph{\(w \times h\)-rectangular grid graph} has vertex set \(\{1,\dots,w\}\times \{1,\dots,h\}\) and two vertices $(x,y)$ and $(x',y')$ are adjacent if $x = x'$ and $|y - y'| = 1$ or $y = y'$ and $|x - x'| = 1$. We call $w$ and $h$ the \emph{width} and \emph{height} of the graph. A graph is a \emph{partial grid graph} if it is a subgraph of some rectangular grid graph. The \emph{height} of a partial grid graph $G$ is the smallest value $h$ such that $G$ is a subgraph of a $w \times h$-rectangular grid graph.

\subparagraph{Graph Width Parameters}

We start with the definition of tree decompositions and path decompositions.

\begin{definition}
    A \emph{tree decomposition} of a graph $G$ is a pair $(T,\{X_t\}_{t\in V(T)})$ consisting of a tree $T$ and a mapping assigning to each node $t\in V(T)$ a set $X_t\subseteq V(G)$ (called \emph{bag}) such that 
    \begin{enumerate}
    \item $\bigcup_{t \in V(T)} X_t = V(G)$
    \item for every edge $uv \in E(G)$, there is a $t \in V(T)$ such that $u,v \in X_t$
    \item for every vertex $v\in V(G)$, the nodes of bags containing $v$ form a subtree of $T$
    \end{enumerate}
    The \emph{width} of a tree decomposition is the maximum size of a bag minus~1. The \emph{\param{treewidth}} of a graph $G$ is the minimal width of a tree decomposition of $G$.
\end{definition}

The terms \emph{path decomposition} and \emph{\param{pathwidth}} are defined accordingly, where the tree $T$ is replaced by a path. We then simply use the ordering $X_1, \dots, X_k$  of the bags in that path to describe the decomposition. 

The \emph{\param{bandwidth}} of a graph $G$ is the smallest integer $k$ such that there exists a vertex ordering $\sigma$ of $G$ where $|\sigma(u) - \sigma(v)| \leq k$ for every edge $uv \in E(G)$. It is easy to see that the \param{treewidth} of a graph is at most its \param{pathwidth} and its \param{pathwidth} is at most its \param{bandwidth}~\cite{kaplan1996pathwidth}. For proper interval graphs, the \param{bandwidth} depends on the \emph{\param{clique number}}, i.e., the size of the largest complete subgraph of the graph. In fact, the following proposition is a direct consequence of Theorem~4.1 in~\cite{{kaplan1996pathwidth}}. 

\begin{proposition}[Kaplan and Shamir~{\cite[Theorem~4.1]{kaplan1996pathwidth}}]\label{prop:proper-bandwidth}
    The \param{bandwidth} of a proper interval graph is equal to its \param{clique number} minus~1.
\end{proposition}

\subparagraph{Ordered Hamiltonian Paths and Cycles}

A \emph{Hamiltonian path} of a graph $G$ is a path that contains all the vertices of $G$. A graph is \emph{traceable} if it has a Hamiltonian path. A \emph{Hamiltonian cycle} of a graph $G$ is a cycle that contains all the vertices of $G$. If a graph contains a Hamiltonian cycle, then the graph is \emph{Hamiltonian}. Here, we only consider \emph{ordered} Hamiltonian paths and cycles. For an ordered Hamiltonian path one of the two possible orderings of the path is fixed. For an ordered Hamiltonian cycle, we fix one start vertex and one of the two possible orderings. Given a partial order $\pi$ on a graph's vertex set, an ordered Hamiltonian path or an ordered Hamiltonian cycle is \emph{$\pi$-extending} if its order is a linear extension of $\pi$. Using these concepts we can define generalizations of the using precedence constraints on the vertices. For \textsc{Hamiltonian Path} this is straightforward.

\begin{problem}{\textsc{Partially Ordered Hamiltonian Path Problem} (POHPP)}
\begin{description}
\item[\textbf{Instance:}] A graph $G$, a partial order $\pi$ on the vertex set of $G$.
\item[\textbf{Question:}]
Is there an ordered Hamiltonian path $(v_1, \dots, v_n)$ in $G$ such that for all $i, j \in \{1,\dots,n\}$ it holds that if $(v_i,v_j) \in \pi$, then $i \leq j$?
 \end{description}
\end{problem}

For \textsc{Hamiltonian Cycle} we define the problem with precedence constraints as follows.

\begin{problem}{\textsc{Partially Ordered Hamiltonian Cycle Problem} (POHCP)}
\begin{description}
\item[\textbf{Instance:}] A graph $G$, a partial order $\pi$ on the vertex set of $G$.
\item[\textbf{Question:}]
Is there an ordered Hamiltonian path $(v_1, \dots, v_n)$ in $G$ such that $v_1$ and $v_n$ are adjacent and for all $i, j \in \{1,\dots,n\}$ it holds that if $(v_i,v_j) \in \pi$, then $i \leq j$?
 \end{description}
\end{problem}

The edge-weighted versions where we look for a Hamiltonian path or cycle with minimum total weight are called MinPOHPP and MinPOHCP.

For \textsc{Hamiltonian Path} and \textsc{Hamiltonian Cycle}, there exist graph classes where one of the problems is trivial while the other is \NP-hard. It is easy to show that \textsc{Hamiltonian Path} is \NP-complete on the graphs that are not Hamiltonian.\footnote{Take a graph $G$ and add a universal vertex $u$ and a leaf that is only adjacent to $u$. The resulting graph is not Hamiltonian and has a Hamiltonian path if and only if $G$ has a Hamiltonian path.} In contrast, \textsc{Hamiltonian Cycle} is \NP-complete on traceable graphs~\cite[Lemma~21.18]{korte2018combinatorial}. For the variants with precedence constraints, the latter result is not possible (unless $\P = \NP$) as we can solve POHCP using an algorithm for POHPP.

\begin{observation}\label{obs:cycle-to-path}
    Let $G$ be a graph with $m$ edges and $\pi$ be a partial order on the vertex set of $G$. We can solve (Min)POHCP for $(G,\pi)$ by solving $\O(m)$ instances of (Min)POHPP on $(G,\pi')$ for some partial orders $\pi'$.
\end{observation}

\begin{proof}
    We consider every pair of adjacent vertices $x$ and $y$ where $x$ is minimal in $\pi$ and $y$ is maximal in $\pi$. Construct the partial order $\pi_{x,y}$ by making $x$ the unique minimal element and $y$ the unique maximal element in $\pi$. Every cost-minimal $\pi$-extending Hamiltonian cycle in $G$ contains a cost-minimal $\pi_{x,y}$-extending Hamiltonian path in $G$ for some pair $(x,y)$. 
\end{proof}

Unless $\P = \NP$, we cannot give the similar result for the reverse direction. This follows from the fact that POHPP is \NP-complete on graphs of \param{pathwidth}~4 while POHCP can be solved in polynomial time on that graph class (see \cref{cor:NPhard_pathwidth,thm:pw4}). However, if we are allowed to add a universal vertex, then we can reduce POHPP to POHCP.

\begin{observation}\label{obs:path-to-cycle}
    Let $G$ be a graph and $\pi$ be a partial order on the vertex set of $G$. Let $G'$ be the graph constructed from $G$ by adding a universal vertex $u$ with zero weight edges and let $\pi'$ be the partial order constructed from $\pi$ by making $u$ the unique maximal element. A $\pi$-extending Hamiltonian path $\cP$ in $G$ is cost-minimal if and only if $\cP \oplus u$ is a cost-minimal $\pi'$-extending Hamiltonian cycle in $G'$.
\end{observation}

Note that \cref{obs:path-to-cycle,obs:cycle-to-path} differ in two ways. First, \cref{obs:cycle-to-path} solves both problems on the same graph while in \cref{obs:path-to-cycle} the graphs of the two problems differ. The second difference concerns the reduction type. The procedure given in \cref{obs:path-to-cycle} is a \emph{Karp reduction} also known as polynomial-time many-one reduction, i.e., we only have to solve (Min)POHCP once to solve (Min)POHPP. In contrast, \cref{obs:cycle-to-path} presents a \emph{Cook reduction} that uses an oracle for (Min)POHCP a polynomial number of times to solve (Min)POHPP. \Cref{obs:cycle-to-path} implies that POHPP does not have a polynomial-time algorithm on some graph class $\G$ if POHCP does not have such an algorithm. However, the concept of \NP-completeness is normally considered using Karp reductions. Therefore, \cref{obs:cycle-to-path} cannot be used to prove \NP-completeness of POHPP on some graph class $\G$ if POHCP is \NP-complete on $\G$. In fact, it is conjectured that the set of problems that are \NP-complete using Karp reductions are a strict subset of those that are \NP-complete using Cook reductions (see, e.g., \cite{lutz1996cook,mandal2014separating}).

\Cref{obs:path-to-cycle} implies the following.

\begin{corollary}\label{cor:path-to-cycle}
    Let $\G$ be a graph class that is closed under the addition of universal vertices. (Min)POHPP is linear-time many-one reducible to (Min)POHCP on $\G$.
\end{corollary}

As mentioned in the introduction, the POHPP was shown to be \NP-complete on the class of complete split graphs~\cite{beisegel2024computing}. As this class is closed under the addition of universal vertices, \cref{cor:path-to-cycle} implies the \NP-completeness of POHCP on that class.

\begin{corollary}
    POHCP is \NP-complete on complete split graphs.
\end{corollary}

\section{Hardness}
In this section we show that POHPP and POHCP are \NP-complete for graphs with bounded bandwidth, pathwidth and treewidth. 
Containement in \NP is obvious and thus we only have to show hardness.
We first show that POHPP and POHCP are \NP-hard for interval graphs of clique number 5 and 6, respectively.
Then we focus on a very structured graph class, namely rectangular grid graphs and show that the problems remain hard if the height is bounded. 
We then extend the result to the MinPOHPP and MinPOHCP, showing \NP-hardness in rectangular grid graphs for even smaller heights.

We will use reductions from two problems that are defined on Boolean formulas. 
The first problem is the well-known 3-SAT problem. The other problem is \textsc{Monotone Weighted 2-SAT} that uses \emph{positive monotone} Boolean formulas, that is, Boolean formulas where all literals are positive. 

\begin{problem}[\textsc{Monotone Weighted 2-SAT}]
~
\begin{description}
\item[\textbf{Instance:}] A positive monotone Boolean formula \(\Phi\) in $2$-CNF, \(k\in \N\)
\item[\textbf{Question:}] Is there a satisfying assignment of \(\Phi\) that sets at most \(k\) variables to \texttt{true}.
\end{description}
\end{problem}

A simple reduction from \textsc{Vertex Cover} shows that \textsc{Monotone Weighted 2-SAT} is \NP-hard~\cite{porschen2007algorithms}.

For both problems, we will assume that \(\Phi\) consists of the variables \(x_1,\dots, x_n\) and the clauses \(c_1,\dots c_m\).
The literals are denoted by \(\ell_i^j\) where \(\ell_i^j\) is the \(j\)-th literal in clause \(c_i\).
Note that for an instance of 3-SAT the literals can be negated.
\subsection{Proper Interval Graphs}

While both \textsc{Hamiltonian Path} and \textsc{Hamiltonian Cycle} can be solved in linear time on interval graphs~\cite{damschke1993paths,keil1985finding}, Beisegel et al.~\cite{beisegel2024computing} showed that the POHPP is \NP-complete an interval graphs of unbounded \param{clique number}. Here, we extend this result to proper interval graphs of bounded \param{clique number}. Note that on proper interval graphs \textsc{Hamiltonian Path} is trivial as a proper interval graph is traceable if and only if it is connected~\cite{bertossi1983finding}.

\begin{theorem}\label{thm:np-interval}
     POHPP is \NP-complete on proper interval graphs of \param{clique number}~5.
\end{theorem}

\begin{proof}
We present a reduction from \textsc{3-SAT} to  POHPP on proper interval graphs of clique number~5. We construct an instance $(G,\pi)$ of  POHPP (see \cref{fig:unit-interval}). We use the following subgraphs for $G$.

\begin{description}
    \item[Start Gadget] The start gadget $S$ consists of the two adjacent vertices $s$ and $s'$.
    \item[Variable Gadget] For each $x_1, \dots, x_n$, there is a variable gadget $X_i$ that consists of the adjacent vertices $x_i$ and $\overline{x}_i$. We call these vertices \emph{variable vertices}.
    \item[Clause Gadget] For each clause $c_i$, there is a clause gadget $C_i$ consisting of the vertices $a^1_i$, $a^2_i$, $\ell^1_i$, $\ell^2_i$, $\ell^3_i$, $b^1_i$, and $b^2_i$. The vertices $\ell^1_i$, $\ell^2_i$ and $\ell^3_i$ form a clique. Similarly $a^1_i$ and $a^2_i$ as well as $b^1_i$ and $b^2_i$ form cliques of size~2. The vertices $a^1_i$ and $b^1_i$ are adjacent to all of $\ell^1_i$, $\ell^2_i$, and $\ell^3_i$. Vertex $a^2_i$ is only adjacent to $\ell^2_i$ and $\ell^3_i$, while $b^2_i$ is only adjacent to $\ell^1_i$ and $\ell^3_i$ (see two examples given in the boxes with rounded-corners in \cref{fig:unit-interval}).
    \item[End Gadget] The end gadget $T$ consists of the two adjacent vertices $t$ and $t'$.
    \item[Backbone] The backbone $B$ is a path consisting of the vertices $\{r_i \mid 0 \leq i \leq n\}$ and $\{u_i,v_i,w_i \mid 1 \leq i \leq m\}$ in the following order: $(r_0, r_1 \dots, r_n, u_1, v_1, w_1, \dots, u_n, v_n, w_n)$.
\end{description}

To complete the construction of $G$, we explain how these gadgets are combined. We order the gadgets in the following way $S, X_1, \dots, X_n, C_1, \dots, C_m, T$. For every gadget, we define \emph{entry vertices} and \emph{exit vertices}. For start, variable and end gadgets, all vertices are entry and exit vertices. For clause gadgets, $a_i^1$ and $a_i^2$ are the entry vertices, while $b_i^1$ and $b_i^2$ are the exit vertices. The exit vertices of a gadget are completely adjacent to the entry vertices of the succeeding gadget. The backbone is connected to all other gadgets in the following way. Vertex $r_0$ is adjacent to $s$, $s'$, $x_1$ and $\overline{x}_1$. Vertex $r_i$ with $1 \leq i < n$ is adjacent to $x_i$, $\overline{x}_i$, $x_{i+1}$, and $\overline{x}_{i+1}$. The vertex $r_n$ is adjacent to $x_n$, $\overline{x}_n$ as well as $a_1^1$ and $a_1^2$. Vertex $u_i$ is adjacent to $a_i^1$, $a_i^2$ as well as $\ell_i^1$, $\ell_i^2$, and $\ell_i^3$. Vertex $v_i$ is adjacent to $\ell_i^1$, $\ell_i^2$, and $\ell_i^3$ as well as $b_i^1$ and $b_i^2$. Vertex $w_i$ is adjacent to $b_i^1$, $b_i^2$, $a_{i+1}^1$, and $a_{i+1}^2$ except from $w_n$ which is adjacent to $b_n^1$, $b_n^2$ as well as $t$ and $t'$.

To prove that $G$ is a proper interval graph, we consider the following vertex ordering:
\begin{align*}
\sigma = (s, s', r_0, x_1, \overline{x}_1, r_1, x_2, \dots, x_n, \overline{x}_n, r_n, a_1^2, a_1^1, u_1, \ell_1^2, \ell_1^3, \ell_1^1, v_1, b_1^1, b_1^2, w_1, a_2^2, a_2^1,\\ \dots, b_m^1, b_m^2, w_m, t, t').
\end{align*}
It can easily be checked that this ordering is a proper interval ordering. Furthermore, the \param{bandwidth} of $\sigma$ is 4. Due to \cref{prop:proper-bandwidth}, the \param{clique number} of $G$ is at most~5.

The partial order $\pi$ on the vertex set of $G$ is the reflexive transitive closure of the relation containing the following constraints:

\begin{enumerate}[(P1)]
    \item $s \prec v$ for every $v \in V(G) \setminus \{s\}$,\label{ui:p1}
    \item $t \prec v$ for every vertex $v$ in the backbone $B$,\label{ui:p2}
    \item $t \prec s' \prec t'$,\label{ui:p3}
    \item $x_i \prec \ell_j^k$ if the $k$-th literal in $c_j$ is $x_i$,\label{ui:p4}
    \item $\overline{x}_i \prec \ell_j^k$ if the $k$-th literal in $c_j$ is $\overline{x}_i$.\label{ui:p5}
\end{enumerate}

\begin{figure}
    \centering
        \begin{tikzpicture}
        \footnotesize

        \begin{scope}

                \draw[rounded corners, fill, color=lightgray!40!white] (4.75,2.1) rectangle (7.25,-1.12);
        \draw[rounded corners, fill, color=lightgray!40!white] (10.25,2.1) rectangle (12.75,-1.12);
        \node[vertex, label=90:$s$] (s) at (0,1) {};
        \node[vertex, label=-90:$s'$] (s') at (0,0) {};
        %\node[vertex, label=0:$z$] (z) at (5.5,3.25) {};
        \node[vertex, label=90:$x_1$] (x1) at (1,1) {};
        \node[vertex, label=-90:$\overline{x}_1$] (nx1) at (1,0) {};
        \node[vertex, label=90:$x_2$] (x2) at (2,1) {};
        \node[vertex, label=-90:$\overline{x}_2$] (nx2) at (2,0) {};
        \node[vertex, label=90:$x_n$] (xn) at (4,1) {};
        \node[vertex, label=-90:$\overline{x}_n$] (nxn) at (4,0) {};

        \draw (s) -- (x1) -- (nx1) -- (s);
        \draw (s') -- (x1) -- (nx1) -- (s') -- (s);
        \draw (x1) -- (x2) --
               (nx2) -- (nx1) -- (x2) -- (nx1) -- (x1) -- (nx2) -- (x2);
        \draw (x2) --+ (0.5, 0);
        \draw (nx2) --+ (0.5, 0);
        \draw (x2) --+ (0.5, -0.5);
        \draw (nx2) --+ (0.5, 0.5);
        
        \node at (3,0.5) {$\dots$};
        \node at (3,2.5) {$\dots$};

        \node at (9.25,0.5) {$\dots$};
        \node at (9.25,2.5) {$\dots$};

        \draw (xn) -- (nxn);
        \draw (xn) --+ (-0.5,0);
        \draw (nxn) --+ (-0.5, 0);
        \draw (xn) --+ (-0.5, -0.5);
        \draw (nxn) --+ (-0.5, 0.5);

        \node[vertex, label=90:$a_1^1$] (a11) at (5,1) {};
        \node[vertex, label=-90:$a_1^2$] (a12) at (5,0) {};

        \draw(xn) -- (a11) -- (a12) -- (nxn) -- (a11) -- (xn) -- (a12);

        \node[vertex, lightgray!50!gray, label=90:$r_0$] (r0) at (0.5,2.5) {};
        \draw (r0) --+ (0.05,-0.35);
        \draw (r0) --+ (0.125,-0.35);
        \draw (r0) --+ (-0.125,-0.35);
        \draw (r0) --+ (-0.05,-0.35);
        \node[vertex, lightgray!50!gray, label=90:$r_1$] (r1) at (1.5,2.5) {};
        \draw (r1) --+ (0.05,-0.35);
        \draw (r1) --+ (0.125,-0.35);
        \draw (r1) --+ (-0.125,-0.35);
        \draw (r1) --+ (-0.05,-0.35);
        \node[vertex, lightgray!50!gray, label=90:$r_2$] (r2) at (2.5,2.5) {};
        \draw (r2) --+ (0.05,-0.35);
        \draw (r2) --+ (0.125,-0.35);
        \draw (r2) --+ (-0.125,-0.35);
        \draw (r2) --+ (-0.05,-0.35);
        
        \draw (r0) -- (r1) -- (r2);
        
        \end{scope}

        \begin{scope}
        \node[vertex, label=90:{$\ell_1^1$}] (x11) at (6,1.5) {};
        \node[vertex, label=-4:$\ell_1^2$] (x12) at (6,0.5) {};
        \node[vertex, label=-90:$\ell_1^3$] (x13) at (6,-0.5) {};
        \node[vertex, label=90:$b_1^1$] (C11) at (7,1) {};
        \node[vertex, label=-90:$b_1^2$] (C12) at (7,0) {};

        \draw (a11) -- (a12);
        \draw (a12) -- (x12);
        \draw (a12) -- (x13);
        \draw (a11) -- (x11);
        \draw (a11) -- (x12);
        \draw (a11) -- (x13);
        
        \draw (x11) -- (x12) -- (x13);
        \draw[bend left] (x13) to (x11);

        \node[vertex, label=90:$a_2^1$] (a21) at (8,1) {};
        \node[vertex, label=-90:$a_2^2$] (a22) at (8,0) {};

        \draw (C11) -- (C12) -- (a21) -- (a22) -- (C11) -- (a21);
        \draw (C12) -- (a22);
        
        \draw (C11) -- (x11);
        \draw (C11) -- (x12);
        \draw[bend angle=15, bend left] (C11) to (x13);
        \draw[bend angle=15, bend right]  (C12) to (x11);
        \draw (C12) -- (x13);
        
        \draw (a21) --+ (0.5,0.25);
        \draw (a21) --+ (0.5,-0.25);
        \draw (a21) --+ (0.5,-0.75);
        \draw (a22) --+ (0.5,-0.25);
        \draw (a22) --+ (0.5,0.25);
        
        \begin{scope}[xshift=-2.5cm]
        \node[vertex, label=90:$a_{m}^1$] (am1) at (13,1) {};
        \node[vertex, label=-90:$a_{m}^2$] (am2) at (13,0) {};
        \node[vertex, label=90:$\ell_{m}^1$] (xm1) at (14,1.5) {};
        \node[vertex, label=-4:$\ell_{m}^2$] (xm2) at (14,0.5) {};
        \node[vertex, label=-90:{$\ell_{m}^3$}] (xm3) at (14,-0.5) {};
        \node[vertex, label=90:$b_m^1$] (cm1) at (15,1) {};
        \node[vertex, label=-90:$b_m^2$] (cm2) at (15,0) {};
        \node[vertex, label=90:$t$] (t) at (16,1) {};
        \node[vertex, label=-90:$t'$] (t') at (16,0) {};

        \node[vertex, lightgray!50!gray, label=90:$w_{m-1}$] (wm-1) at (12.5,2.5) {};
        \draw (wm-1) --+ (0.05,-0.35);
        \draw (wm-1) --+ (0.125,-0.35);
        \draw (wm-1) --+ (-0.125,-0.35);
        \draw (wm-1) --+ (-0.05,-0.35);
        \node[vertex, lightgray!50!gray, label=90:$u_m$] (um) at (13.5,2.5) {};
        \draw (um) --+ (-0.05,-0.35);
        \draw (um) --+ (-0.125,-0.35);
        \draw (um) --+ (0.125,-0.35);
        \draw (um) --+ (0.0825,-0.35);
        \draw (um) --+ (0.04167,-0.35);
        \node[vertex, lightgray!50!gray, label=90:$v_m$] (vm) at (14.5,2.5) {};
        \draw (vm) --+ (0.05,-0.35);
        \draw (vm) --+ (0.125,-0.35);
        \draw (vm) --+ (-0.125,-0.35);
        \draw (vm) --+ (-0.0825,-0.35);
        \draw (vm) --+ (-0.04167,-0.35);
        \node[vertex, lightgray!50!gray, label=90:$w_m$] (wm) at (15.5,2.5) {};
        \draw (wm) --+ (0.05,-0.35);
        \draw (wm) --+ (0.125,-0.35);
        \draw (wm) --+ (-0.125,-0.35);
        \draw (wm) --+ (-0.05,-0.35);

        \draw (wm-1) -- (um) -- (vm) -- (wm);
        \end{scope}
  
        \draw (am1) --+ (-0.5,0.0);
        \draw (am1) --+ (-0.5,-0.5);
    
        \draw (am2) --+ (-0.5,-0.0);
        \draw (am2) --+ (-0.5,0.5);
        
        \draw (am1) -- (xm1);
        \draw (am1) -- (xm2);
        \draw (am1) -- (xm3);
        \draw (am2) -- (xm2);
        \draw (am2) -- (xm3);
        \draw (am2) -- (am1);

        \draw (xm1) -- (xm2) -- (xm3);
        \draw[bend left] (xm3) to (xm1);

        \draw (cm1) -- (xm1);
        \draw (cm1) -- (xm2);
        \draw[bend angle=15, bend left] (cm1) to (xm3);
        \draw[bend angle=15, bend right] (cm2) to (xm1);
        \draw (cm2) -- (xm3);

        \draw (cm2) -- (cm1) -- (t) -- (cm2) -- (t') -- (t);
        \draw (cm1) -- (t');

        \node[vertex, lightgray!50!gray, label=90:$r_{n-1}$] (rn-1) at (3.5,2.5) {};
        \draw (rn-1) --+ (0.05,-0.35);
        \draw (rn-1) --+ (0.125,-0.35);
        \draw (rn-1) --+ (-0.125,-0.35);
        \draw (rn-1) --+ (-0.05,-0.35);
        \node[vertex, lightgray!50!gray, label=90:$r_n$] (rn) at (4.5,2.5) {};
        \draw (rn) --+ (0.05,-0.35);
        \draw (rn) --+ (0.125,-0.35);
        \draw (rn) --+ (-0.125,-0.35);
        \draw (rn) --+ (-0.05,-0.35);
        \node[vertex, lightgray!50!gray, label=90:$u_1$] (u1) at (5.5,2.5) {};
        \draw (u1) --+ (-0.05,-0.35);
        \draw (u1) --+ (-0.125,-0.35);
        \draw (u1) --+ (0.125,-0.35);
        \draw (u1) --+ (0.0825,-0.35);
        \draw (u1) --+ (0.04167,-0.35);
        \node[vertex, lightgray!50!gray, label=90:$v_1$] (v1) at (6.5,2.5) {};
        \draw (v1) --+ (0.05,-0.35);
        \draw (v1) --+ (0.125,-0.35);
        \draw (v1) --+ (-0.125,-0.35);
        \draw (v1) --+ (-0.0825,-0.35);
        \draw (v1) --+ (-0.04167,-0.35);
        \node[vertex, lightgray!50!gray, label=90:$w_1$] (w1) at (7.5,2.5) {};
        \draw (w1) --+ (-0.05,-0.35);
        \draw (w1) --+ (-0.125,-0.35);
        \draw (w1) --+ (0.125,-0.35);
        \draw (w1) --+ (0.05,-0.35);
        \node[vertex, lightgray!50!gray, label=90:$u_2$] (u2) at (8.5,2.5) {};
        \draw (u2) --+ (-0.05,-0.35);
        \draw (u2) --+ (-0.125,-0.35);
        \draw (u2) --+ (0.125,-0.35);
        \draw (u2) --+ (0.0825,-0.35);
        \draw (u2) --+ (0.04167,-0.35);
        
        \draw (rn-1) -- (rn) -- (u1) -- (v1) -- (w1) -- (u2);

        \draw[very thick, color=variables] (s) -- (x1)--(nx2) --+ (0.5, 0.5);
        \draw[very thick, color=variables] (nxn) --+(-0.5,0);
        \draw[very thick, color=variables] (nxn) -- (a11);
        \draw[very thick, color=clauses](a11) -- (x12)-- (C11) --(a21)  --+ (0.5,0.25);
        \draw[very thick, color=clauses] (am1) --+(-0.5,0);
        \draw[very thick, color=clauses] (am1) -- (xm3);
        \draw[very thick, color=clauses,bend angle=15, bend right] (xm3) to (cm1);
        \draw[very thick, color=clauses] (cm1) --(t);

        \draw[very thick, color=backbone](t)--(wm)--(vm)--(um)--(wm-1);
        \draw[very thick, color=backbone](u2)--(w1)--(v1) -- (u1) -- (rn) -- (rn-1);
        \draw[very thick, color=backbone](r2)--(r1)--(r0) -- (s');
        \draw[very thick, color=variables2](s')--(nx1) -- (x2) --+ (0.5,0);
        \draw[very thick, color=variables2](xn)--+(-0.5,0);
        \draw[very thick, color=variables2](xn)--(a12);
        \draw[very thick, color=clauses2](a12)--(x13);
        \draw[very thick, color=clauses2, bend left](x13) to (x11);
        \draw[very thick, color=clauses2, bend angle = 15, bend left](x11) to (C12);
        \draw[very thick, color=clauses2](C12)--(a22)  --+ (0.5,0.25);
        \draw[very thick, color=clauses2] (am2) --+(-0.5,0);
        \draw[very thick, color=clauses2](am2)--(xm2) --(xm1);
        \draw[very thick, color=clauses2, bend angle = 15, bend left](xm1) to (cm2);
        \draw[very thick, color=clauses2](cm2)--(t');
        
        \end{scope}
    \end{tikzpicture}
    \caption{Complete construction of the proof of \cref{thm:np-interval}. The boxes mark the clause gadgets. Gray vertices are forced to be to the right of $t$ in $\pi$.
    The thick path is a Hamiltonian path. The blue part visits the variable vertices whose corresponding literals are satisfied, the red part traverses through the vertices of satisfied literals, the orange part goes back to $s'$, and the light and dark green parts visit the remaining variable and clause vertices, respectively.}
    \label{fig:unit-interval}
\end{figure}
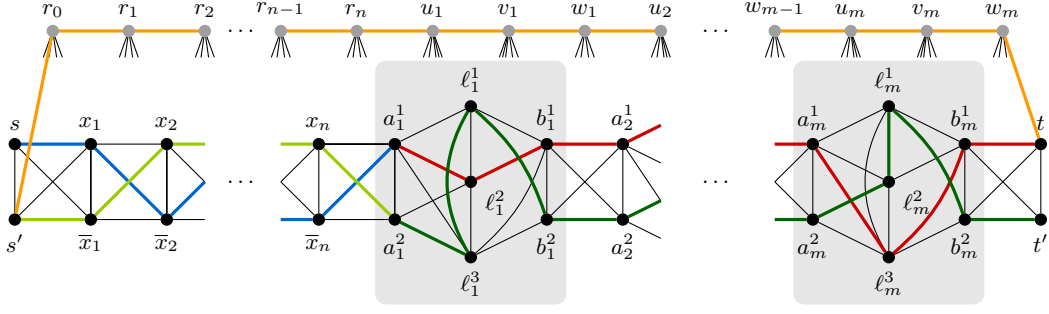

It remains to show that the instance of POHPP described above has a $\pi$-extending Hamiltonian path if and only if the given formula \(\Phi\) has a satisfying assignment.

\proofsubparagraph{From a $\pi$-extending path to a satisfying assignment} Let $\cP$ be a $\pi$-extending Hamiltonian path of $G$. Due to \pef{ui:p1}, the path $\cP$ starts with $s$. Let $\cP_{s,t}$ be the subpath of $\cP$ starting in $s$ and ending in $t$. %The following claim shows how we can define an assignment of $\Phi$.

\begin{claim}
    For any $i \in \{1,\dots,n\}$, $\cP_{s,t}$ contains exactly one of the two vertices $x_i$ and $\overline{x}_i$.  
\end{claim}

\begin{claimproof}
    Due to \pef{ui:p2}, the path $\cP_{s,t}$ cannot contain any of the vertices in the backbone. Any path between $s$ and $t$ that does not contain those vertices has to contain at least one of the variable vertices $x_i$ and $\overline{x}_i$ for every $i \in \{1,\dots,n\}$.

    Let $\cP_{t,s'}$ be the subpath of $\cP$ between $t$ and $s'$ and $\cP_{s',t'}$ be the subpath of $\cP$ between $s'$ and $t'$. Due to \pef{ui:p1} and \pef{ui:p3}, vertex $s$ is to the left of $t$ in $\cP$, $t$ is to the left of $s'$ in $\cP$ and $s'$ is to the left of $t'$ in $\cP$. Therefore, $\cP_{s,t}$, $\cP_{t,s'}$, and $\cP_{s',t'}$ do not share an inner vertex. All three paths have to use one of the vertices $x_i$, $\overline{x}_i$ and $r_i$. Hence, $\cP_{s,t}$ cannot contain both $x_i$ and $\overline{x}_i$.
\end{claimproof}

Due to this claim, we can define an assignment $\A$ of $\Phi$ in the following way: A variable $x_i$ is set to true in $\A$ if and only if the vertex $x_i$ is contained in $\cP_{s,t}$. Since $\cP_{s,t}$ cannot contain any of the $u_i$ or $v_i$, it is clear that $\cP_{s,t}$ must contain at least one of the vertices $\ell_i^1$, $\ell_i^2$, and $\ell_i^3$ for any $i \in \{1,\dots,m\}$. As those vertices can only be visited by $\cP$ if their respective variable vertex is to the left of them in $\cP$, it follows that $\A$ is an satisfying assignment of $\Phi$.

\proofsubparagraph{From a satisfying assignment to a $\pi$-extending Hamiltonian path} Let $\A$ be a satisfying assignment of $\Phi$. We define the ordered Hamiltonian path $\cP$ as follows (see color scheme in \cref{fig:unit-interval}). We start in $s$. Then we successively visit the variable vertices in such way that we visit $x_i$ if the variable $x_i$ is set to true in $\A$ and otherwise we visit $\overline{x}_i$ (blue). Afterwards, we visit $a_1^1$. As $\A$ is a satisfying assignment, there is at least one of the vertices $\ell_1^1$, $\ell_1^2$, and $\ell_1^3$ that can be visited as the next vertex. We choose one of these vertices as next vertex and then visit $b_1^1$. We repeat this process for each $i \in \{2, \dots, m\}$. Afterwards, we visit vertex $t$ (red). As the backbone vertices were only restricted by \pef{ui:p2}, we now can use the complete backbone in decreasing order to reach vertex $s'$ (orange). Next we visit all the remaining variable vertices in increasing order (light green).
In the clause gadget for \(c_1\), we first visit $a_1^2$. 
Then, we visit all previously unvisited literal vertices in an order that visits \(\ell_1^2\) or \(\ell_1 ^3\) first and \(\ell_1^1\) or \(\ell_1^3\) last, followed by \(b_1^2\).
We repeat this procedure for all $i \in \{2, \dots, m\}$ and finally visit $t'$ (dark green). The resulting path is a $\pi$-extending Hamiltonian path of $G$.
\end{proof}

\Cref{obs:path-to-cycle} implies directly the \NP-hardness of the POHCP on interval graphs of \param{clique number}~6. In contrast, proper interval graphs are not closed under the addition of universal vertices. However, we can slightly adapt the proof of \cref{thm:np-interval} to extend the cycle result to proper interval graphs

\begin{theorem}\label{thm:np-interval-cycle}
    POHCP is \NP-complete on proper interval graphs of \param{clique number}~6.
\end{theorem}

\begin{proof}
    We construct the graph $G'$ by adding another backbone path $B'$ to the graph $G$ constructed in the proof of \cref{thm:np-interval} by introducing for every vertex $x$ of $B$ a copy $x'$. That is, $B'$ contains the vertices  $\{r'_i \mid 0 \leq i \leq n\}$ and $\{u'_i,v'_i,w'_i \mid 1 \leq i \leq m\}$ in the following order: $(r'_0, r'_1 \dots, r'_n, u'_1, v'_1, w'_1,  \dots,\allowbreak u'_n, v'_n, w'_n)$. 

    The vertex $x' \in B'$ has the same neighbors in $G - B$ as the vertex $x$. Furthermore, vertex $x'$ is adjacent to $x$ and to the successor of $x$ in $B$.

    We first claim that the resulting graph is a proper interval graph of \param{clique number}~6. To this end, we consider the following vertex ordering:
    \begin{align*}
    \sigma = (s, s', r_0, r'_0, x_1, \overline{x}_1, r_1, r'_1, x_2, \dots, x_n, \overline{x}_n, r_n, r'_n, a_1^2, a_1^1, u_1, u'_1, \ell_1^2, \ell_1^3, \ell_1^1, v_1, v'_1, \\ b_1^1, b_1^2, w_1, w'_1, a_2^2, a_2^1, \dots, b_m^1, b_m^2, w_m, w'_m, t, t').
    \end{align*}
    The ordering $\sigma$ forms a proper interval ordering of \param{bandwidth}~5. Similar as in \cref{thm:np-interval}, this proves that the graph is a proper interval graph of \param{clique number}~6, due to \cref{prop:proper-bandwidth}.

    For every element $x'$ of $B'$, we add a constraint to $\pi$ that forces $x'$ to be to the right of $t'$. The proof of the correctness of the reduction follows the same lines as the proof of \cref{thm:np-interval}. To construct a $\pi$-extending Hamiltonian cycle, we use the path $B'$ in reversed order to go from $t'$ to $s$.
\end{proof}

Combining the last two theorems with \cref{prop:proper-bandwidth}, we can conclude the following.

\begin{theorem}\label{cor:NPhard_pathwidth} The following statements are true.
\begin{enumerate}
    \item POHPP is \NP-complete on graphs of \param{bandwidth}, \param{pathwidth} or \param{treewidth}~$4$.
    \item POHCP is \NP-complete on graphs of \param{bandwidth}, \param{pathwidth} or \param{treewidth}~$5$.
\end{enumerate}
\end{theorem}

We will show in \cref{sec:algo} that these results are tight (if $\P \neq \NP$) for at least \param{bandwidth} and \param{pathwidth} by showing that POHPP can be solved in polynomial time on graphs of \param{pathwidth}~3 and POHCP can be solved in polynomial time on graphs of \param{pathwidth}~4. This tightness can be formulated even stronger as adding one or two edges is enough to let the problems' complexity switching from easy to hard.

\begin{corollary}
    For every $k,\ell \in \N$, let $\B^\ell_k$ be the class of graphs $G$ that contain a set $F$ of $\ell$ edges such that $G - F$ has \param{bandwidth}~$k$. Then it holds:
    \begin{enumerate}
    \item POHPP is \NP-complete on $\B^1_3$.
    \item POHCP is \NP-complete on $\B^1_4$ and on $\B^2_3$.
\end{enumerate}
\end{corollary}

\begin{proof}
    We replace the backbone $B$ in the proof of \cref{thm:np-interval} by an edge from $t$ to $s'$. The correctness of the reduction follows analogously. If we remove this edge, then the remaining graph is a proper interval graph of \param{clique number}~4 and, thus, has \param{bandwidth}~3.

    Similarly, we can replace the backbone $B'$ in the proof of \cref{thm:np-interval-cycle} by an edge from $t'$ to $s$ and removing that edge leads to proper interval graph of \param{clique number}~5. If we use both edge replacements, we get an element of $\B^2_3$.
\end{proof}

\subsection{Rectangular Grid Graphs}
We now build on the ideas presented in \cref{thm:np-interval} to show various results for hardness in rectangular grid graphs. 
All gadgets used in the following reductions are \(w\times h\) rectangular grid graphs. 
For such a gadget \(Z\), we use \(Z[a,b]\) to denote the vertices in the \(a\)-th row and the \(b\)-th column of the gadget. Note that in the figures the rows are counted from top to bottom and the columns from left to right.

\subsubsection{Unweighted Graphs}
\begin{theorem}\label{thm:np-grid}
   POHPP is \NP-complete on rectangular grid graphs of height $7$.
\end{theorem}

\begin{proof}
    We present a reduction from 3-SAT to  POHPP.
    Let \(\Phi\) be a 3-SAT formula over variables \(x_1,\dots, x_n\) with clauses \(c_1,\dots, c_m\).
    We construct an instance \((G,\pi)\) for  POHPP where \(G\) is a \((6n + 4m + 2)\times 7\) rectangular grid graph.
    In order to be able to argue about different parts of \(G\), it is conceptually subdivided into different gadgets. We also name some special vertices within these gadgets.
   
     \begin{description}
        \item[Left border gadget \(S\) (rose):] A \( 2\times 7\) grid graph. 
        \item[Variable gadget \(X_i\) (yellow and blue):] A \(6\times 7\) grid graph for \(i=1,\dots, n\).
        \(X_i\) is assigned to variable \(x_i\).
        We call the vertices \(X_i[3,3]\) and \(X_i[3,4]\) the \emph{negative variable vertices} of \(X_i\). The vertices \(X_i[5,4]\) and \(X_i[5,5]\) are the \emph{positive variable vertices}. 
        We denote the middle four columns of \(X_i\) by \(X_i'\).
        \item[Clause gadget \(C_j\) (green and light blue):]  A \(4\times 7\) grid graph  assigned to \(c_j\) for \(j=1,\dots m\).
        The vertices \(C_j[a+2,2]\) and  \(C_j[a+2,3]\) are called \emph{literal vertices} and are assigned to the literal \(\ell_j^a\) for \(a=1,2,3\).
    \end{description}
    We denote the vertex \(X_1[3,1]\) as \(s\) and the vertex \(C_m[3,4]\) as \(t\).
    \(G\) is made up of the gadgets in the following order: \(S,X_1,\dots, X_n, C_1,\dots, C_m\), see \cref{fig:grid_gadgets}.
    The partial order \(\pi\) is the reflexive, transitive closure of the  following constraints:
    
    \begin{enumerate}[(P1)]
   \item $s\prec v$ for all \(v\in V(G) \setminus \{s\}\). \label{item:grid_start}
        \item $t\prec v$ for all \(v\in S\).
        \item $t\prec v$ for all vertices \(v\) in Rows~\(1,2,6\) and \(7\).
        \item $t\prec v$ if \(v\) is \(X_i[a,b]\) for \((a,b)\in \{4\}\times \{2,3,4,5\}\).
        \item $u\prec v$ if \(u\) is a negative variable vertex in \(X_i\) and \(v\) in \(C_j\) is assigned to a literal \(\overline{x}_i\).\label{item:grid_lit1}
        \item $u\prec v$ if \(u\) is a positive variable vertex in \(X_i\) and \(v\) in \(C_j\) is assigned to a literal \(x_i\).\label{item:grid_lit2}
    \end{enumerate}
    
    We call \pef{item:grid_lit1} and \pef{item:grid_lit2} \emph{literal constraints}.
    \begin{figure}
        \centering
        \includegraphics[width=\textwidth]{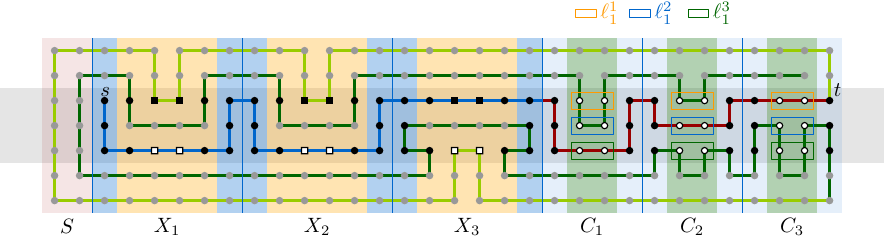}
        \caption{The gadgets for \cref{thm:np-grid} combined to represent a formula. The gray vertices come after \(t\) in the partial order \(\pi\).
        A black square is negative variable vertex, a white square is a positive variable vertex.
        A white disk marks a literal vertex.
        }
        \label{fig:grid_gadgets}
    \end{figure}

Now we show that the instance described above has a $\pi$-extending Hamiltonian path if and only if the given formula \(\Phi\) is satisfiable.

\proofsubparagraph{From a $\pi$-extending Hamiltonian path to a satisfying assignment}

    Let \(\mathcal{P}\) be a \(\pi\)-extending Hamiltonian path of \(G\).
     By \pef{item:grid_start}, \(\cP\) starts in \(s\).
    Let \(\mathcal{P}_{s,t}\) be the prefix of \(\mathcal{P}\) ending in \(t\).
    Observe that by the definition of \(\pi\), all vertices that are in Rows \(1,2,6\) or \(7\) or in Row~4 of some \(X_i'\) cannot lie on \(\mathcal{P}_{s,t}\).

    \begin{claim}\label{claim:pos_neg}
    For any $i \in \{1,\dots,n\}$, $\cP_{s,t}$ either contains both positive or both negative variable vertices of \(X_i\).
\end{claim}
\begin{claimproof}
Assume for a contradiction that there is an \(i\), such that at least one of \ \(X_i[3,3]\) or \(X_i[3,4]\) and one of \(X_i[5,3]\) or \(X_i[5,4]\) lie on \(\cP_{s,t}\).
As rows \(1,2,4,6,7\) of \(X_i'\) are all after \(t\) in any linear extension of \(\pi\), a valid prefix path \(\cP_{s,t}\) cannot switch rows within \(X_i'\).
In particular, it can only cross \(X_i'\) in one row.
\(\cP_{s,t}\) thus first visits either the positive or the negative variable vertices and then exits \(X_i'\) on the right and then reenters it from the right, see \cref{fig:grid-assignment} for an illustration.
After traversing \(X_i'\) twice, all vertices of \(X_i'\) that can be visited before \(t\) are already on the path.
Thus, there are no possible vertices left to cross \(X_i'\) again and the path cannot reach \(t\), a contradiction.
\end{claimproof}
\begin{figure}
    \centering
      \includegraphics[page=2]{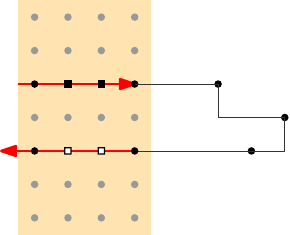}
    \caption{The path \(\cP_{s,t}\) is stuck when traversing both row \(3\) and row \(5\) of \(X_i\)}
    \label{fig:grid-assignment}
\end{figure}

\cref{claim:pos_neg} allows us to define an assignment of the variables of \(\Phi\) as follows.
Set \(x_i\) to false if and only if \(\cP_{s,t}\) traverses \(X_i'\) through the negative variable vertices.
Consider a clause gadget \(C_j\) and a vertex \(v\) in column \(2\) or \(3\) of \(C_j\) that is visited by \(\cP_{s,t}\).
As only vertices in  rows \(3,4\) or \(5\) are on \(\cP_{s,t}\), the vertex \(v\) is a literal vertex.
Let \(\ell_j^a\) be the literal assigned to \(v\).
If \(\ell_j^a = x_i\) for some \(i\), then by \pef{item:grid_lit2}, the positive variable vertices of \(x_i\) are before \(v\) on \(\cP_{s,t}\). 
Thus the assignment specified above sets \(x_i\) to true, satisfying \(c_j\).
In the other case, \(\ell_j^a = \overline{x}_i\), by \pef{item:grid_lit1}, the negative variable vertices of \(x_i\) are visited before \(v\) on \(\cP_{s,t}\). 
As in this case \(x_i\) is set to false, the literal \(\ell_j^a\) satisfies \(c_j\).

\begin{figure}
    \centering
    \includegraphics[page=2, width=\textwidth]{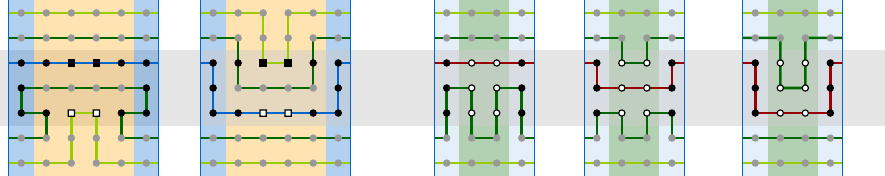}
    \caption{The path \(\cP\) in the variable and clause gadgets}
    \label{fig:grid_variableClauses}
\end{figure}

\proofsubparagraph{From a satisfying assignment to a $\pi$-extending Hamiltonian path}
Assume that \(\Phi\) has a satisfying assignment. 
We now show that there is a \(\pi\)-extending Hamiltonian path \(\cP\) in the graph.
We describe the path in three parts.
The first part, \(\cP_{s,t}\), is a valid prefix that connects \(s\) and \(t\).
This part is marked in light blue and red in \cref{fig:grid_gadgets,fig:grid_variableClauses}.
The second part (\(\cP_2\)) of the path visits the remaining variable vertices.
This part is drawn in light green in the figures.
In the last part (\(\cP_3)\), all remaining vertices, in particular those in the clause gadgets are visited (see the dark green path in the figures).

The path \(\cP_{s,t}\) starts in \(s\). While visiting the variable gadgets,  it maintains the invariant, that the first and last vertex visited in each gadget are \(X_i[3,1]\) and \(X_i[3,6]\), respectively.
If \(x_i\) is set to false, then it stays in row~\(3\) and crosses \(X_i\) through the negative variable vertices.
In the other case, it visits the vertices $X_i[3,1]$, $X_i[4,1]$, and $X_i[5,1]$  and then crosses \(X_i'\) through the positive variable vertices and then takes the three vertices in the last column to reach \(X_i[3,6]\). See the blue path in \cref{fig:grid_gadgets}.

A similar process is repeated for the clause gadgets. The path is drawn red in \cref{fig:grid_gadgets}. Again, the first and last vertex visited in \(C_i\) are \(C_i[3,1]\) and \(C_i[3,4]\).
Let \(\ell_j^a\) be an arbitrary literal that satisfies the clause \(c_j\), then \(\cP_{s,t}\) crosses the middle two columns of \(C_j\) via the literal vertices assigned to \(\ell_j^a\).
The first and last column are used to reach this row from \(C_j[3,1]\) and then reach \(C_j[3,4]\).
As \(t=C_m[3,4]\), \(\cP_{s,t}\)  ends in the last column of the graph.

To see that \(\cP_{s,t}\) is the prefix of a \(\pi\)-extending Hamiltonian path, we first note that only vertices \(v\) that are not constrained to be right of \(t\) are visited.
Furthermore, if the path visits a vertex in a clause gadget that is assigned to a negative literal, then the corresponding negative variable vertices were visited before. 
Analogously, all positive variable vertices are visited before vertices assigned to a positive literal.
This makes sure that the literal constraints involving the visited literal vertices are satisfied.

For the subpath \(\cP_2\), all vertices in variable gadgets that are not on \(\cP_{s,t}\) are visited.
This is done by first walking along row \(1\) with the invariant, that each gadget is entered and left in this row.
The path only leaves row \(1\) in a variable gadget if \(x_i\) is set to true  in the satisfying assignment. In this case, it visits the negative variable vertices as indicated in the second figure in \cref{fig:grid_variableClauses}.
When reaching \(S[1,1]\) the path goes down to \(S[7,1]\) and then continues in row \(7\), again only leaving the row to visit positive variable vertices if \(x_i\) is set to false, see the leftmost figure in \cref{fig:grid_variableClauses}.

In the final phase, all remaining vertices are visited. Recall that this part of the path is called \(\cP_3\) and it is drawn in dark green in all figures.
For the clause gadgets, \(\cP_3\) first visits all vertices below \(\cP_{s,t}\) maintaining the invariant that each gadget is entered and left in row \(6\). Then the path visits the remaining vertices in \(S\) and then all remaining vertices above \(\cP_{s,t}\) with the invariant that every gadget is entered and left in row \(2\).
The path ends in \(C_m[2,3]\).
See \cref{fig:grid_gadgets} and the three figures at the right of \cref{fig:grid_variableClauses} for an illustration that this invariant can be maintained for all possible assignments of the variables and true literals.

By maintaining the invariants mentioned above, it is clear, that the path \(\cP\) that results from combining \(\cP_{s,t}, \cP_2\) and \(\cP_3\) is indeed a Hamiltonian path.
It is left to argue that the path \(\cP\) is indeed \(\pi\)-extending.
We have already argued that \(\cP_{s,t}\) is a prefix of a \(\pi\)-extending Hamiltonian path.
Therefore, the first four types of constraints of $\pi$ as well as the literal constraints for all variable vertices on \(\cP_{s,t}\) are fulfilled.
To see that the remaining literal constraints are satisfied, we argue as follows. If some vertex assigned to a literal is visited on~\(\cP_3\), then all remaining positive and negative variable vertices have already been visited in~$\cP_2$.
\end{proof}

\begin{theorem}
    POHPP is \NP-complete on rectangular grid graphs of height \(h\) for every fixed \(h \geq 7\).
\end{theorem}
\begin{proof}
The same construction as in the proof of \Cref{thm:np-grid} can be used with a slight modification.
All gadgets are now \(k \times h\) grids.
Gadget widths for the variable and clause gadgets are the same as in the construction of \cref{thm:np-grid}. The left end gadget size is increased to have three columns and the start vertex is now \(S[5,3]\). 
To the right a \(3 \times h\) right end gadget \(T\) is added.

The vertices in Rows~\(8,\dots, h\) take over the ordering constraints of the vertices in Row~\(7\) in the same column.
The proof then follows by the same arguments as that of \cref{thm:np-grid}, replacing the part of \(\cP_2\)  that visits the vertices in Row~\(7\) by a space filling curve that enters and leaves each variable block and each remaining gadget in Row~\(h\).
See \cref{fig:grid_larger_seven} for an illustration, especially of the routing through $S$ and $T$. \qedhere
\begin{figure}
    \centering
    \includegraphics[page=3, width=\textwidth]{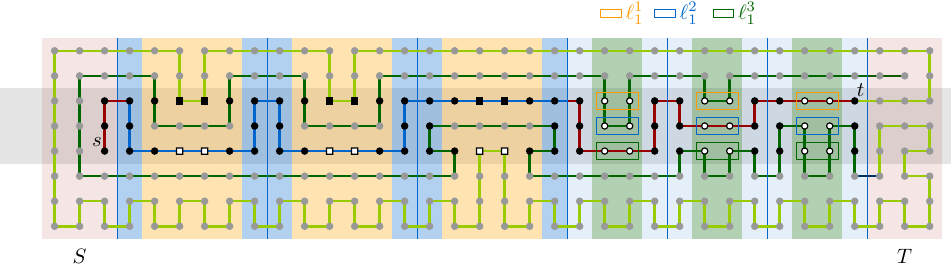}
    \caption{Example for the construction with height \(8\)}
    \label{fig:grid_larger_seven}
\end{figure}
\end{proof}

We now extend the construction above to show that POHCP is also \NP-hard. The construction is a simple extension of that used for \cref{thm:np-grid}. The main difference is, that we have to introduce some more rows and columns in order to be able close the cycle.

\begin{theorem}\label{thm:cycle_grid_unweighted}
    POHCP is \NP-complete on rectangular grid graphs of height $9$.
\end{theorem}

\begin{proof}

We, once again, adapt the proof of \cref{thm:np-grid}.
The main challenge to go from Hamiltonian paths to Hamiltonian cycles is that we have to route the path back to \(s\).
In order to be able to do this, the variable and clause gadgets have to be slightly modified:
  \begin{description}
        \item[Left border gadget \(S\) (rose):] A \(2\times 9\) grid graph. 
        \item[Right border gadget \(T\) (rose): ] A \(2\times 9\) grid graph. 
        \item[Variable gadget \(X_i\) (yellow and blue):] A \(8\times 9\) grid graph for \(i=1,\dots, n\).
        \(X_i\) is assigned to variable \(x_i\).
        We call the vertices \(X_i[4,4]\) and \(X_i[4,5]\) the \emph{negative variable vertices} of \(X_i\). The vertices \(X_i[6,4]\) and \(X_i[6,5]\) are the \emph{positive variable vertices}. 
        We denote the middle four columns of \(X_i\) by \(X_i'\).
        \item[Clause gadget \(C_j\) (green and light blue):]  A \(6\times 9\) grid graph  assigned to \(c_j\) for \(j=1,\dots m\).
        The vertices \(C_j[a+3,3]\) and  \(C_j[a+3,4]\) are assigned to the literal \(\ell_j^a\).
    \end{description}

 We denote the vertex \(X_1[7,1]\) as \(s\) and the vertex \(T[4,2]\) as \(t\).
    \(G\) is made up of the gadgets in the following order: \(S,X_1,\dots, X_n, C_1,\dots, C_m, T\), see \cref{fig:grid_gadgets}.
    The partial order \(\pi\) is the reflexive, transitive closure of the  following constraints:
    \begin{enumerate}[(P1)]
   \item $s\prec v$ for all \(v\in V(G) \setminus \{s\}\). \label{item:cycle_grid_start}
        \item $t\prec v$ for all \(v\in S\cup T\).
        \item $t\prec v$ for all vertices \(v\) in Rows~\(1,2,3,8\) and \(9\).
        \item $t\prec v$ if \(v\) is \(X_i[a,b]\) for \((a,b)\in \{5,6\}\times \{3,4,5,6\}\).
        \item $t\prec v$ if \(v \in C_j[a,b]\) with \((a,b)\in \{7\}\times\{2,3\}\)
        \item $u\prec v$ if \(u\) is a negative variable vertex in \(X_i\) and \(v\) in \(C_j\) is assigned to a literal \(\overline{x}_i\). \label{item:cycle_grid_lit1}
        \item $u\prec v$ if \(u\) is a positive variable vertex in \(X_i\) and \(v\) in \(C_j\) is assigned to a literal \(x_i\).\label{item:cycle_grid_lit2}
    \end{enumerate}
   Again, we call \pef{item:cycle_grid_lit1} and \pef{item:cycle_grid_lit2} \emph{literal constraints}.
    \begin{figure}
        \centering
        \includegraphics[width=\textwidth, page=5]{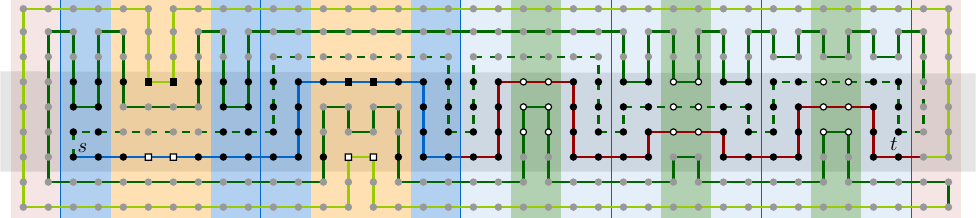}
        \caption{The gadgets for \cref{thm:cycle_grid_unweighted} combined to represent a formula. The gray vertices come after \(t\) in the partial order \(\pi\).
        A black square is negative variable vertex, a white square is a positive variable vertex.
        A white disk marks a literal vertex.
        }
        \label{fig:grid_gadgets_cycle}
    \end{figure}

The proof for the correctness follows the same structure as that in \cref{thm:np-grid}. We will only highlight the parts that have to be adapted for the new gadgets.
To show that a \(\pi\)-extending Hamiltonian path in this graph implies the existence of a satisfying assignment, we have to argue that \cref{claim:pos_neg} still holds and that \(\cP_{s,t}\) visits a literal vertex in each clause gadget.
As in all \(X_i'\) each row except for the ones that contain variable gadgets cannot be used, the proof of \cref{claim:pos_neg} carries over. Furthermore, the only vertices in the middle two columns of a clause gadget that can be visited are the literal vertices. 
With these two observations, the proof of \cref{thm:np-grid} carries over and using the same rule for the assignment yields a satisfying assignment.

For the other direction the main difference to \cref{thm:np-grid} is that the path has to visit all vertices in a proper order and then reach back to \(s\).
We again subdivide the path into \(\cP_{s,t}, \cP_{2}\) and \(\cP_{3}\). 
\(\cP_{s,t}\) is defined as above, with the difference that the path now changes between the gadgets in row \(7\) instead of row \(3\).
\(\cP_2\) again stays in the top and bottom row, only leaving the row to visit the variable vertices not visited yet by \(\cP_{s,t}\).
Finally, all remaining vertices have to be visited. Note that as \(t\) was already visited and all variable vertices are visited, the partial order does not give any more constraints on these vertices. 
\(\cP_3\) now visits all remaining vertices, except for those that are directly above the vertices visited by \(\cP_{s,t}\). This part is marked solid dark green in \cref{fig:grid_gadgets_cycle}. Finally, \(\cP_3\) closes the cycle by visiting the remaining vertices (dashed dark green blue curve).
\cref{fig:grid_gadgets_cycle} shows all possible assignments of variables and clauses and how the paths can be routed in each case.
\end{proof}

\subsubsection{Weighted Graphs}
For the following proofs, we will use  \textsc{Monotone Weighted 2-SAT} in the reduction.

\begin{theorem}\label{thm:minPOHPPheight5}
    MinPOHPP is \NP-hard on rectangular grid graphs of height $5$.
\end{theorem}

\begin{proof}
We follow a similar strategy as for the proof for \cref{thm:np-grid,thm:cycle_grid_unweighted}, however we use \textsc{Monotone Weighted 2-SAT} instead of \textsc{3-SAT} for the reduction. 
Another relevant difference is that there will be no negative variable vertices.
The variable gadgets \(X_i\) are  \(6\times 5\) grids and the clause gadgets \(Y_i\) are \(7\times 5\) grids. In addition there is a \(1\times 5\) right boundary gadgets \(T\).
Let \(X_i\) be a variable gadget. Then \(X_i[3,3]\) is the \emph{variable vertex} for \(x_i\).
Similarly, for a clause gadget \(C_j\), the vertex \(C_j[4,2]\) is assigned to the first literal in clause \(c_j\) and \(C_j[4,4]\) is assigned to the second literal.
Let \(s=X_1[5,1]\) and \(t=T[5,1]\).
The partial order is the reflexive and transitive closure of the following constraints:
\begin{enumerate}[(P1)]
\item \(s\prec v\) for all \(v\in V(G)\setminus \{s\}\) \label{enum:weighted:sstart}
    \item \(t \prec v\) for all \(v\in X_i[a,b]\) with \(a\in \{1,2\}\) and for \(X_i[3,2]\) and \(X_i[4,2]\)
    \item \(t \prec v\) for all \(v\in C_j[a,b]\) with \(a\in \{1,2,3\}\)
    \item \( v \prec t\) for all \(v=C_j[4,3]\). \label{enum:weighted:forcev}
    \item \(X_i[3,3]\prec C_j[4,2a]\) if \(\ell_j^a = x_i\). \label{enum:weighted:posliteral}
\end{enumerate}
The edge weights are chosen as follows: All weights are set to \(0\) except for the edges \(\{X_i[4,3], X_i[3,3]\}\), which have weight~\(1\).
The gadgets are combined in the order \(X_1,\dots,X_n,\allowbreak C_1, \dots, C_m, T\).

We now show that there is a \(\pi\)-extending Hamiltonian Path of weight at most \(k\) in this graph if and only if there is a satisfying assignment for \(\Phi\) that sets at most \(k\) variables to true.

\proofsubparagraph{From a $\pi$-extending Hamiltonian Path to a satisfying assignment}
Let \(\cP\) be a \(\pi\)-extending Hamiltonian path with weight at most \(k\). By \pef{enum:weighted:sstart}, \(s\) is the first vertex on \(\cP\). Let  \(\cP_{s,t}\) be the prefix that connects \(s\) to \(t\).
Consider \(\cP_{s,t}\) in a clause gadget \(C_j\).
The vertex \(C_j[4,3]\) has to be on \(\cP_{s,t}\) by \pef{enum:weighted:forcev}. Furthermore, as \(C_j[4,3]\) has only three neighbors that can be on \(\cP_{s,t}\) one of the literal vertices is also on \(\cP_{s,t}\).

Consider the variable vertex \(v\) in \(X_i\). The vertices in the column left and the row above of \(v\) in \(X_i\) are not on \(\cP_{s,t}\). Thus, if \(v\) is on \(\cP_{s,t}\), the vertical edge below the variable vertex was used in the gadget.

The assignment of the variables can be chosen as follows. Variable \(x_i\) is set to true if and only if \(\cP_{s,t}\) visits the variable vertex in \(X_i\).
The constraints of \pef{enum:weighted:posliteral} now ensure that a clause gadget can only be crossed if at least one literal is true in the assignment.
As the total weight of the path is at most \(k\), there are at most \(k\) variable vertices that are visited and thus the assignment is a valid solution to the \textsc{2-SAT} instance.

\proofsubparagraph{From a satisfying assignment to a $\pi$-extending Hamiltonian Path}
The path is again subdivided into three parts.
 See \cref{fig:minPOHPP-height-5} for an illustration.
 The first part enters and leaves each gadget in the last row. If \(x_i\) is set to true, it visits the variable vertex \(X_i[3,3]\) and its horizontal neighbor \(X_i[3,4]\) in between. 
 If the variable is set to false, the vertices \(X_i[4,3]\) and \(X_i[4,4]\) are visited. 
The path then visits the clause gadgets, visiting a literal vertex that represents a literal that satisfies this clause as well as \(C_j[4,3]\).

The middle part of the path simply goes up in \(T\), traverses through row \(1\) and then goes back down to \(S[4,1]\).
The path then visits the remaining vertices, first in the variable and then in clause gadgets. The path always enters and leaves a gadget in row \(4\). In particular, the remaining variable vertices are visited before the remaining literal vertices, making sure that \pef{enum:weighted:posliteral} is satisfied.
Again, we refer to \cref{fig:minPOHPP-height-5} for an illustration.
\begin{figure}
    \centering
    \includegraphics[page=6]{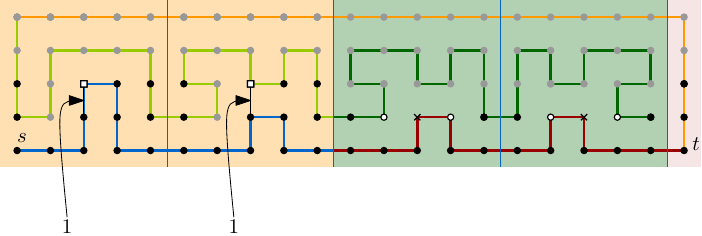}
    \caption{The construction for MinPOHPP in grids of height at most \(5\). Gray vertices are after \(t\) in the partial order. Vertices marked with an \(\times\) have to be visited before \(t\). The  variable vertices are white squares and literal vertices are white circles.}
    \label{fig:minPOHPP-height-5}
\end{figure}
\end{proof}

\begin{theorem}
   MinPOHCP is NP-hard for rectangular grid graphs of height \(6\).
\end{theorem}
\begin{proof}
We adapt the construction in the proof of \cref{thm:minPOHPPheight5} in order to be able to close the cycle.
The variable gadgets are now \(6\times 6\) rectangular grids and the clause gadgets are \(7\times 6\) rectangular grid. In addition to the \(1\times 6\) right end gadget \(T\), there is also a \(1\times 6\) left end gadget \(S\).
Let \(s=S[6,1]\) and \(t=T[6,1]\).
The variable vertex for \(x_i\) is the vertex \(X_i[4,3]\) and the literal vertices for \(\ell_j^1\) and \(\ell_j^2\) are \(C_j[5,3]\) and \(C_j[5,5]\), respectively.
The only weighted edges are \(\{X_i[4,3],X_i[5,3]\}\) with a weight of \(1\), all other edges have weight \(0\).
The graph \(G\) is then obtained by concatenating the gadgets in the order \(S, X_1,\dots, X_n, C_1, \dots C_m, T\). 

The partial order is essentially the same as the one defined in \cref{thm:minPOHPPheight5}. The bottom five rows of \(G\) are identified with the height \(5\) graph defined above, where the last column in a variable gadget and the first and last column of the clause gadgets have the same constraints as their neighboring rows in the same gadget. The top row of \(G\) simply takes over the constraints from row \(2\).

Going from a \(\pi\)-extending Hamiltonian cycle of weight at most \(k\) to a satisfying assignment works the same way as in \cref{thm:minPOHPPheight5} since the subgraph of vertices that can be visited before \(t\) is defined analogously to that used above and, thus, the argument carries over.

Given a satisfying assignment, a \(\pi\)-extending Hamiltonian cycle can be found as indicated in \cref{fig:grid-skizze}.
Te part between \(s\) and \(t\) first visits all variable vertices assigned to true variables and can then cross the clause gadgets.
The path goes back to \(S\) and then visits the remaining variable vertices and simply crosses the clause gadgets. Finally, it visits the remaining literal vertices and connects back to \(s\).
\end{proof}

As grid graphs of height $h$ have \param{outerplanarity} $\lceil \frac h2 \rceil$, we can state the following complexity results for this parameter.

\begin{figure}
    \centering
    \includegraphics[page=7, width=\textwidth]{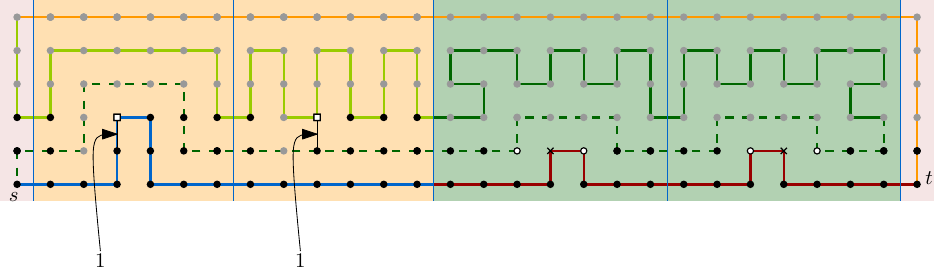}
    \caption{The reduction for MinPOHCP. White circles are the variable vertices, the white squares are the literal vertices. Vertices marked with \(\times\) are before \(t\) in \(\pi\), gray vertices are after \(t\) in~\(\pi\).}
    \label{fig:grid-skizze}
\end{figure}

\begin{corollary}
The following statements are true.
\begin{enumerate}
    \item POHPP is \NP-complete on graphs of \param{outerplanarity}~$4$.
    \item POHCP is \NP-complete on graphs of \param{outerplanarity}~$5$.
    \item MinPOHPP and MinPOHCP are \NP-hard on graphs of \param{outerplanarity}~$3$.
\end{enumerate}
\end{corollary}

Note that for graphs of outerplanarity~1, i.e., outerplanar graphs, MinPOHPP can be solved in polynomial time~\cite{beisegel2024computing}. MinPOHCP is trivial for outerplanar graphs as they have at most one Hamiltonian cycle and this cycle can be found in polynomial time~\cite{syslo1979characterizations}.

\section{Algorithms}\label{sec:algo}

In the section above, we showed that POHPP and POHCP are \NP{}-complete for graphs with \param{treewidth} and \param{pathwidth} at least four or five, respectively. 
In this section we give polynomial-time algorithms for \param{pathwidth} at most three or four and \param{treewidth} at most two or three. In both cases, we will present an algorithm for the cycle case. Using \cref{obs:path-to-cycle}, we then can directly derive an algorithm for the path case.

\subsection{Pathwidth}\label{subsec:pathwidth3}

Let \(X_1,\dots, X_k\) be a path decomposition of width \(4\) of a graph \(G\).
For our algorithm we assume that every bag contains exactly five vertices. Furthermore, we assume that \(X_i\) and \(X_{i+1}\) differ in exactly two vertices.
We call the unique vertex \(u\in X_{i}\setminus X_{i+1}\) the vertex that is \emph{forgotten} in \(X_{i+1}\) and the unique vertex \(w\in X_{i+1}\setminus X_i\) the vertex that is \emph{introduced} in \(X_{i+1}\).
Furthermore, let \(G_i\) be the induced subgraph on \(V_i = \bigcup_{j=1}^i X_j\).
A general path-decomposition of width~4 can be computed in linear time~\cite{bodlaender1996linear-time,furer2016faster}.
It might not fulfill the constraint that \(|X_i|= 5\) and that consecutive bags differ in exactly two vertices. It can however be brought easily into this form in linear time.

As each vertex set \(X_i\) is a separator of \(G\), the following observation holds:
 \begin{observation}\label{obs:forgotten_connected}
 The vertex \(u\) forgotten in \(X_{i+1}\) has no edge to a vertex in \(V\setminus V_i\).
 The vertex \(w\) introduced in \(X_{i+1}\) has no edge to a vertex in \(V_i  \setminus X_{i+1}\).
 \end{observation}

 Our algorithm is based on the folklore dynamic programming algorithm that solves \textsc{Hamiltonian Cycle} for graphs of bounded \param{pathwidth}.
 We will first sketch that algorithm and then describe the modifications needed when considering partial order constraints.

Let \(\C\) be an ordered Hamiltonian cycle in \(G\) and \(\C_i = \C \cap G_i\) be the set of paths (or, in case $V_i = V(G)$, the set containing $\C$) induced by \(V_i\). 
For \(P=(v_1,\dots, v_\ell)\in \C_i\) we call \(v_1\) the \emph{start vertex}, \(v_\ell\) the \emph{end vertex} and the remaining vertices the \emph{interior vertices}. 
In some cases, we refer to the start and end vertices as \emph{terminal vertices}.

If \(|P|=1\), we call \(P\) a \emph{trivial path}.
 A cycle \(\C\) induces a \emph{signature} that encodes the interaction of \(\C_i\) with the vertices in \(X_i\).
 The subproblems in the dynamic program are then all pairs of a bag \(X_i\) and a signature \(\sigma\) for that bag and the entry for that pair is \texttt{true} if there is a set of paths that visit all vertices in \(V_i\) and interact with \(X_i\) as indicated by the signature and \texttt{false} otherwise.
 The recurrence in the dynamic program then considers all entries for \(X_{i-1}\) with a signature \(\gamma\) that is \emph{compatible} to $\sigma$, i.e., both $\gamma$ and $\sigma$ can hypothetically be induced by the same Hamiltonian cycle.
 If there is one such signature $\gamma$ with value \texttt{true} and the vertex introduced in \(X_i\) can be connected to the vertices in \(X_i \cap X_{i-1}\) to form \(\sigma\), then the entry of $\sigma$ is \texttt{true}.

 The number of signatures and pairs of compatible signatures is bounded by a function depending on the width of the bags and, thus, the algorithm is an \FPT{} algorithm when parameterized by \param{pathwidth}.

When not considering ordering constraints, the information if there is a set of paths that visit every vertex in \(V_i\) exactly once and that interact with \(V\setminus V_i\) as indicated by \(\sigma\) is enough.
However, in POHCP the order in which the vertices are visited is relevant. 
 Thus, we also need information about the partition of the vertices in \(X_i\) to the paths.
 One could naively extend the dynamic program above to store all possible partitions of the vertices in \(V_i\) to the paths. 
 However, as the number of these partitions grows exponentially with~\(n\) this approach is not feasible in general.
For \param{pathwidth}~\(4\), however, we can show that the relevant information can be maintained and computed efficiently.

The following lemmas are the basis for our algorithm.
\begin{lemma}\label{lem:pw3_storinternal}
    Let \(\C\) be a Hamiltonian cycle, let \(1\leq i < k\) and \(u\) be the vertex forgotten in \(X_{i+1}\). Then
    \begin{enumerate}
        \item \(u\) is an interior vertex of a path in \(\C_i\), and
        \item \(\C_i\) contains at most two trivial paths.
    \end{enumerate}
\end{lemma}
\begin{proof}
As $u$ must have two neighbors in $\C$ and, by \cref{obs:forgotten_connected}, it cannot be connected to vertices later, it has to be an interior vertex. 

If $u$ is an interior vertex, then there are at least two further vertices that are part of the same path in $\C_i$. Thus, there are at most two vertices left that can form trivial paths.
\end{proof}

Unlike in the case without constraints, for the POHCP we have to distinguish two different types of non-trivial paths in $\C_i$. Let $\lambda = (v_1, \dots, v_n)$ be the linear extension of $\pi$ that is contained in $\C$. Then, one of the paths of $\C_i$ might contain $v_n$. This path could also contain $v_1$ but not all vertices between $v_1$ and $v_n$ in $\lambda$, i.e., the ordering of the path is not part of a linear extension of $\pi$. Therefore, these paths have to be handled differently than paths that do not contain the end of $\lambda$. We call a path $P$ in $\C_i$ a \emph{close-part} if it contains the end vertex of $\lambda$. Otherwise, we call it a \emph{midpart}. Note that a path $P =(w_1, \dots, w_\ell)$ can only be a close-part if there is a $j \in \{1, \dots, \ell\}$ such that $(w_1, \dots, w_j)$ forms a suffix of a linear extension of $\pi$ and $(w_{j+1}, \dots, w_\ell)$ forms a prefix of a linear extension of $\pi$.   

\begin{lemma}\label{lem:pw3_form}
     For \(1\leq i < k\), \(\C_i\) has one of the following forms:
    \begin{enumerate}
        
    \item \(\C_i\) contains exactly one non-trivial path together with at most two trivial paths.
    \item \(\C_i\) contains two non-trivial midparts of $\C$.
    \item \(\C_i\) contains a non-trivial midpart and a non-trivial close-part of $\C$.
\end{enumerate}
Furthermore, \(\C_i = \{\C\}\) if and only if \(i=k\).
\end{lemma}
\begin{proof}
If \(\C_i\) contains only one non-trivial path, then there are at most two trivial paths, due to \cref{lem:pw3_storinternal}. \(\C_i\) cannot contain three non-trivial paths as they would need six terminal vertices but there are only five vertices in $X_i$. Now assume that \(\C_i\) contains two non-trivial paths. Due to \cref{lem:pw3_storinternal}, there is at least one interior vertex in $X_i$. Hence, there cannot be a trivial path in $\C_i$.

If $\C_i$ contains two close-parts, then one of them contains the start vertex $s$ of the linear extension $\lambda$ of $\C$ and the other the end vertex $t$ of $\lambda$. However, in $\C$ these two vertices are adjacent. Therefore, the edge $st$ is present in $V_i$ and the two paths in $\C_i$ must form one path.

If \(\C_i = \{\C\}\), then \(V_i=V(G)\) has to hold. This is only the case for \(V_k\) and thus the last part of the lemma holds.
\end{proof}

The following straight-forward observation helps us to reduce the number of relevant partitions of the vertices to the paths.
\begin{observation}\label{obs:change}
    For each Hamiltonian cycle \(\C\) such that \(\C_{i}\) contains two non-trivial paths, there is an index \(\ell \leq i\) such that all \(\C_j\) for \(j=\ell,\dots, i\) contain two non-trivial paths and either \(\ell=1\) or \(\C_{\ell-1}\) contains only one non-trivial path.
\end{observation}

\begin{figure}
    \centering
    \includegraphics{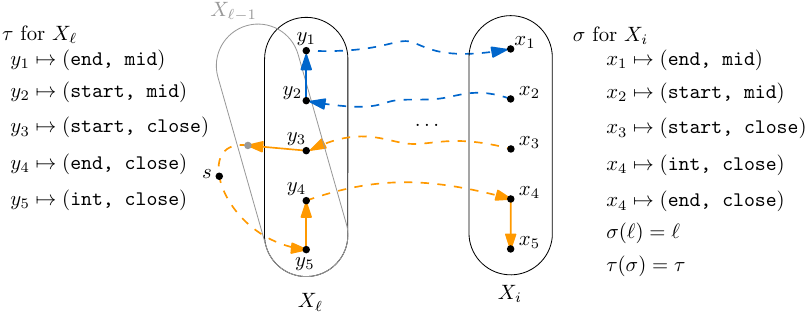}
    \caption{Illustration of the definition of a signature for a fixed directed Hamiltonian cycle \(\C\).}
    \label{fig:signature_example}
\end{figure}

We define the \emph{signature} \(\sigma\) of a bag \(X_i\) as a mapping of the vertices in \(X_i\) to the possible types of vertices and paths. 
Intuitively, a signature \(\sigma\) encodes the following information for a path that is a candidate for a \(\pi\)-extending Hamiltonian cycle: How does the path interact with \(X_i\)? If there are two non-trivial paths in $X_i$, then in which bag \(X_\ell\) did the second path appear? And, finally, how does the path interact with \(X_{\ell}\)?
If \(\sigma\) induces two paths, then \(\sigma\) also contains a value \(\ell\) for \(1\leq \ell \leq i\). If \(\ell\neq i\), then we also store a signature \(\tau\) for \(X_{\ell}\).
Let \(\C(\sigma)\) be the set of paths induced by \(\sigma\). For every vertex $v \in X_i$, we store the following values: $\sigma_p \in \{1,2,3\}$ encoding in which of the paths the vertex $v$ is contained in and \(\sigma_t(v)\in \{\texttt{start},\texttt{end},\texttt{int}\}\) encoding the vertex type assigned to \(v\) by \(\sigma\). Furthermore, we store for every path whether it is a midpart (\texttt{mid}) or a close-part (\texttt{close}) (see \cref{fig:signature_example} for an illustration of these concepts).

We call a signature \emph{valid} if \(\C(\sigma)\) and \(\sigma_t\) do not violate \cref{lem:pw3_storinternal} and \cref{lem:pw3_form}.
In the following we will write \(\ell(\sigma)\) and \(\tau(\sigma)\) for the values \(\ell\) and \(\tau\) stored with a signature \(\sigma\). 

 We say that the signatures \(\gamma\) for \(X_{i-1}\) and \(\sigma\) for \(X_{i}\) are \emph{compatible} if they are both valid and there is a way to extend the paths induced by \(\gamma\) with the vertex \(w\) introduced in \(X_{i}\) to get the signature \(\sigma\).
 If \(\sigma\) induces two paths and \(\ell(\sigma) < i\), then \(\sigma\) and \(\gamma\) are only compatible if \(\ell(\sigma) = \ell(\gamma)\) and \(\tau(\sigma) = \tau(\gamma)\) for \(\ell(\sigma) \leq i-2\) and \(\tau(\sigma) = \gamma = \tau(\gamma)\) if \(\ell(\sigma) = i-1\).
 Intuitively this means that the second path appeared in the same bag and that bag has the same signature for both \(\sigma\) and \(\gamma\).
Furthermore, if \(\ell(\sigma)=i\), then a signature \(\gamma\) is only compatible to \(\sigma\) if \(\gamma\) induces only one non-trivial path.

\begin{lemma}\label{lem:pw3_onlyone}
Let \(\sigma\) be a signature for \(X_i\) such that \(\C(\sigma)\) contains two non-trivial paths. 
Let \(\beta\) and $\gamma$ be signatures of \(X_{i+1}\) that are compatible to \(\sigma\) such that \(\C(\beta)\) and \(\C(\gamma)\) both contain two non-trivial paths. Then $\ell(\beta) = \ell(\gamma)$, $\tau(\beta) = \tau(\gamma)$. Furthermore, $\beta_p(v) = \gamma_p(v)$, and $\beta_t(v) = \gamma_t(v)$ for all $v \in X_{i+1}$. In other words, $\beta$ and $\gamma$ may only differ in the mapping of the paths to \texttt{mid} and \texttt{close}.
\end{lemma}
\begin{proof}
Consider the bags \(X_{i}, X_{i+1}\) and \(X_{i+2}\).
Note that \(X_{i+2}\) exists as \(\C(\gamma)\) contains two paths and such a signature is not valid for~\(X_k\), due to \cref{lem:pw3_form}.

See \cref{fig:pw3_onlyone}  for an illustration of the following arguments. Let \(w\) be the vertex introduced in \(X_{i+1}\) and \(u\) be the vertex forgotten in \(X_{i+1}\).
Furthermore, let \(u'\) be the vertex forgotten in \(X_{i+2}\).
Then, by \cref{lem:pw3_storinternal}, \(\sigma_t(u) = \texttt{int}\) as well as \(\beta_t(u') = \gamma_t(u') = \texttt{int}\).
 All other vertices in \(X_{i+1}\) have to be mapped to terminal vertices to fulfill the assumption on $\beta$ and \(\gamma\) and, thus, \(u'\) is the unique vertex that is assigned to be interior by $\beta$ and \(\gamma\).
 Similarly, \(u\) is the only vertex that is assigned to be interior by \(\sigma\).

Furthermore, note that \(\beta_t(w) \neq \texttt{int}\) and \(\gamma_t(w)\neq \texttt{int}\) since connecting \(w\) to two vertices in \(X_{i+1}\setminus X_{i}\) would either close a cycle, if they are in the same path, or it would connect the paths, a contradiction to the assumption on $\beta$ and \(\gamma\).
Thus, \(\beta_t(w), \gamma_t(w)\in \{\texttt{start}, \texttt{end}\}\) and it follows that \(w\neq u'\).
As \(\gamma_t(u') = \texttt{int}\)  but \(\sigma_t(u')\in \{\texttt{start}, \texttt{end}\}\), in any path that has signature \(\sigma\) in \(X_{i}\) and induces two non-trivial paths in \(X_{i+1}\), vertex \(w\) is connected to \(u'\).
Without loss of generality, assume \(\sigma_t(u') = \texttt{start}\) and $\sigma_p(u') = 1$. Then \(\beta_t(w) = \gamma_t(w) = \texttt{start}\), \(\beta_t(u') = \gamma_t(u')= \texttt{int}\), $\beta_p(w) = \gamma_p(w) = \beta_p(u') = \gamma_p(u') = 1$ and for all other vertices in \(v\in X_i \cap X_{i+1} \cap X_{i+2}\) we have \(\beta_t(v) = \gamma_t(v) = \sigma_t(v)\) as well as \(\beta_p(v) = \gamma_p(v) = \sigma_p(v)\).
Thus $\beta$ and $\gamma$ have the same mappings \(\beta_t\) and \(\gamma_t\).
As \(\ell(\sigma)  = \ell(\gamma)\) and \(\tau(\sigma) = \tau(\gamma)\) holds for every pair of compatible signatures, the statement follows.
 \qedhere
 \begin{figure}
     \centering
     \includegraphics{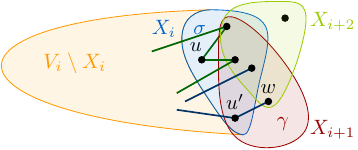}
     \caption{Illustration of the proof of \cref{lem:pw3_onlyone}. Vertex \(u'\) is a terminal vertex in the signature \(\sigma\) of $X_i$ but an interior vertex in the signature \(\gamma\) of $X_{i+1}$. Thus, it is connected to \(w\).}
     \label{fig:pw3_onlyone}
 \end{figure}
\end{proof}

\begin{lemma}\label{lem:pw3_midotherunique}
   Let \(i\geq 2\) and let \(\sigma\) be a valid signature for \(X_i\) that induces two non-trivial paths with \(\ell(\sigma) < i\).
   Then there is at most one signature \(\gamma\) for \(X_{i-1}\) that is compatible with~$\sigma$.
\end{lemma}
\begin{proof}
If there was more than one signature \(\gamma\) in $X_{i-1}$ that is compatible to \(\sigma\), then
there are at least two sequences \(\sigma =\sigma_i^1, \sigma_{i-1}^1,\dots\allowbreak, \sigma_{\ell(\sigma)+1}^1,  \sigma_{\ell(\sigma)}^1 = \tau(\sigma)\) and \(\sigma = \sigma_i^2, \sigma_{i-1}^2,\dots,\allowbreak\sigma_{\ell(\sigma)+1}^2,  \sigma_{\ell(\sigma)}^2 = \tau(\sigma)\), such that $\sigma_j^1$ is compatible with $\sigma_{j-1}^1$ and $\sigma_j^2$ is compatible with $\sigma_{j-1}^2$ for all $j$. Recall that a signature is only compatible if either the information about \texttt{close} and \texttt{mid} matches or if a vertex that can be the first or last vertex of a linear extension is introduced in this bag.

Applying \cref{lem:pw3_onlyone} with \(i= \ell(\sigma)\) implies that \(\sigma_{\ell(\sigma)+1}^1 = \sigma_{\ell(\sigma)+1}^2\).
Applying this inductively yields \(\sigma_j^1 = \sigma_j^2\) for all \(j\) with \(\ell(\sigma)\leq j \leq i-1\) and the lemma follows.
 \end{proof}

Given a signature \(\sigma\) for \(X_i\), a \emph{path mapping} \(f_\sigma\) assigns each vertex \(v\in V_i\) to one of the paths induced by \(\sigma\).
In particular \(f_\sigma(v) = \sigma_p(v)\) for \(v\in X_i\), i.e., the vertices in \(X_i\) are mapped according to \(\sigma\).

\begin{lemma}\label{obs:pw:unique_partition}
    Let \(\sigma\) be a signature for \(X_i\). There is at most one possible path mapping \(f_\sigma\).
\end{lemma}

\begin{proof}
    We use an induction over $i$. First note that the mapping of the vertices of $X_i$ to the paths in $\sigma$ is fixed by $\sigma$. This implies directly that the statement holds for $i = 1$. Thus, assume that $i > 1$. If $\sigma$ induces only one non-trivial path, then every vertex of $V_i \setminus X_i$ is contained in that path. As mentioned above, the mapping of the vertices in $X_i$ is fixed by~$\sigma$.

    So assume $\sigma$ induces two non-trivial paths. If $\ell(\sigma) = i$, then one of the paths was created in $X_i$, i.e., this path contains the vertex introduced in $X_i$. The other path has to contain all the vertices of $V_i \setminus X_i$.

    So assume that $\ell(\sigma) < i$. Due to \cref{lem:pw3_midotherunique}, there is at most one signature $\sigma'$ of $X_{i-1}$ that is compatible with $\sigma$. By induction, the path mapping of $\sigma'$ is unique if it exists. Hence, for all vertices in $V_i \setminus X_i$, the path mapping is fixed. As above, the path mapping of the vertices in $X_i$ is directly fixed by $\sigma$.
\end{proof}

A path mapping \emph{contradicts} \(\pi\) if there is no possible order in which the vertices  in \(V\setminus V_i\) can be appended and prepended to the paths without violating the constraints given by \(\pi\).

Now we can define the dynamic program that will solve POHCP:
For  each bag \(X_i\) and each valid signature  \(\sigma\) for \(X_i\), define the subproblem \(D[i,\sigma]\).
If there is a path mapping \(f_\sigma\) for \(X_i\) that does not contradict \(\pi\) then \(D[i,\sigma]=f_\sigma\), otherwise \(D[i,\sigma]=\bot\).
The algorithm considers increasing values of \(i\).
In the base case \(D[1,\sigma]\), we check for the one possible path partition \(f_\sigma\) whether the vertices in the paths can be ordered according to $\pi$.

For an entry \(D[i,\sigma]\) with \(i\geq 2\), iterate over all valid signatures \(\gamma\) for \(X_{i-1}\) that are compatible to \(\sigma\) and where \(D[i-1,\gamma]\neq \bot\).
Let \(w\) be the vertex introduced in \(X_{i}\) and \(f_\sigma^\gamma\) be the path partition that extends \(D[i-1,\gamma]\) by assigning \(w\) to the path defined by \(\sigma\).
As \(\sigma\) and \(\gamma\) are consistent, the position of \(w\) in the paths is uniquely determined. Iterate over the vertices \(V\setminus \{w\}\) to explicitly check if adding \(w\) in this unique position contradicts \(\pi\). As $w$ could also join two paths, we describe the details by considering the following: We assume that we append a path $P_1$ (that might only contain $w$) to the front of another path $P_2$ (that might only contain $w$). Note that we can split the joining of two paths into two appending operations. 

\begin{itemize}
    \item If both $P_1$ and $P_2$ are midparts, then we check whether no vertex of $P_1$ is a successor of a vertex of  $P_2$ in $\pi$.
    \item If $P_1$ is a midpart and $P_2$ is a close-part, then we check whether all successors of vertices of $P_1$ in $\pi$ are elements of $P_1 \cup P_2$.
    \item If $P_1$ is a close-part and $P_2$ is a midpart, then we check whether all predecessors of vertices of $P_2$ in $\pi$ are elements of $P_1 \cup P_2$.
\end{itemize}

If there is a \(\gamma\) such that \(w\) can be added to the path, then set \(D[i,\sigma]=f_\sigma^\gamma\).
In the other case, set \(D[i,\sigma] = \bot\).
If there is a valid signature \(\sigma\) for \(X_k\) such that \(D[k,\sigma]\neq \bot\), then the algorithm returns \emph{yes}, in the other case, it returns \emph{no}.

\begin{theorem}\label{thm:pw4}
    MinPOHCP on graphs of \param{pathwidth} at most \(4\) can be solved in \(\O(n^4)\) time.  
\end{theorem}

\begin{proof}
 As in each bag one vertex is forgotten, there are \(\O(n)\) possible bags.
 For each bag, there are \(\O(n)\) valid signatures as there are only $\O(1)$ valid choices of $\C(\sigma)$, at most $\O(n)$ choices for $\ell(\sigma)$ and $\O(1)$ choices for $\tau(\sigma)$.

 There are two cases for these signatures of $X_i$. First, we consider those signatures $\sigma$ where $\C(\sigma)$ contains two non-trivial paths and $\ell(\sigma) < i$. There are $\O(n)$ of these signatures for a bag $X_i$. By \cref{lem:pw3_midotherunique}, there is only one signature $\gamma$ of \(X_{i-1}\) for which we have to check if \(f_\sigma^\gamma\) contradicts \(\pi\). 
 Finding  \(\gamma\) takes \(\O(n)\) time. Checking if it contradicts \(\pi\) takes an additional \(\O(n)\) time. We only have to find out where the predecessors and successors of the newly introduced vertex $w$ lie in the path mapping. Thus, we need a total of \(\O(n)\) time to compute \(D[i,\sigma]\). As there are $\O(n)$ many of these signatures, we can compute their $D$-values in $\O(n^2)$ total time.

 Now let us consider the other signatures. There are only $\O(1)$ of them for a bag $X_i$. There might be $\O(n)$ valid signatures of $X_{i-1}$ that are compatible with $\sigma$ where $\O(1)$ have only one non-trivial path while $\Omega(n)$ might have two non-trivial paths. As mentioned above, the newly introduced vertex $w$ might join the two non-trivial paths induced by the signature of $X_{i-1}$. Then checking for consistency with $\pi$ might take $\O(n^2)$ time as we have to check whether there is a pair of vertices, one of each path, whose tuple in $\pi$ contradict the ordering of the newly build path. Thus, we need $\O(n^3)$ time in total to compute the $D$-values of one such signature $\sigma$ of $X_i$. As there are $\O(1)$ of them, we need $\O(n^3)$ time to compute the $D$-values of all these signatures of $X_i$.

 Thus, processing one bag costs $\O(n^3)$ time. As there are $\O(n)$ bags, our algorithm needs $\O(n^4)$ time in total.

The correctness of the algorithm follows directly by the validity of the steps described above.

 The algorithm described above can be easily modified to store the weight that stems from a path partition, together with the partition. 
 If there are multiple ways to construct the same path partitions, the one that induces the smallest weight is stored.
\end{proof}

As adding a universal vertex to a graph increases the \param{pathwidth} of the graph by exactly one, \cref{obs:path-to-cycle} implies an algorithm for the path variant on graphs of \param{pathwidth} at most $3$.

\begin{theorem}\label{thm:pw3}
    MinPOHPP on graphs of \param{pathwidth} at most \(3\) can be solved in \(\O(n^3)\) time.  
\end{theorem}

As the \param{pathwidth} of a partial grid graph of height $h$ is at most $h$, we can state the following. 

\begin{corollary}
    The following statements are true.
    \begin{enumerate}
        \item MinPOHPP can be solved in $\O(n^3)$ time on partial grid graphs of height at most \(3\).
        \item MinPOHCP can be solved in $\O(n^3)$ time on partial grid graphs of height at most \(4\).
    \end{enumerate}
\end{corollary}

Besides these algorithmic results, our arguments above also imply some structural result about Hamiltonian cycles in graphs of \param{pathwidth} at most $4$. In particular, \cref{obs:pw:unique_partition} implies the following.

\begin{corollary}\label{corol:number-path-pw}
    Let $G$ be an $n$-vertex graph of \param{pathwidth} at most $4$ and let $X_1, \dots, X_k$ be a path decomposition of $G$. Let $\mathfrak{C}$ be the set of all Hamiltonian cycles in $G$. The size of the family of path mappings $\{\C \cap V_i \mid \C \in \mathfrak{C}, i \in \{1,\dots,k\}, V_i = \bigcup_{j=1}^i X_j\}$ is $\O(n^2)$. 
\end{corollary}

Note that this result cannot be extended to the case of Hamiltonian paths even if the \param{bandwidth} of the graph is bounded by~$3$. Consider the graph given in \cref{fig:path-example} and a Hamiltonian path starting in $s$ and ending in $t$. Depending on the choice which of $v_i$ and $w_i$ we add to the prefix or to the suffix of the path, we get an exponential number of path mappings. The same example also shows that the number of Hamiltonian cycles of graphs of \param{bandwidth}~3 is exponential. Similarly, one can show that \cref{corol:number-path-pw} does not hold for graphs of \param{bandwidth}~5. To this end, consider the graph $G$ consisting of the vertex set $V(G) := \{v_1, \dots, v_n\}$ and the edges $E(G) := \{v_iv_j : |i - j| \leq 5\}$. Start with the two paths $v_3,v_1,v_5$ and $v_4,v_2,v_6$. Next, we have to attach one of the vertices $v_7$ and $v_8$ to the terminal vertex $v_3$. However, it is not fixed which of the vertices we have to attach to $v_3$. Repeating this argument, we can show that the number of path mappings is exponential in the number of vertices.

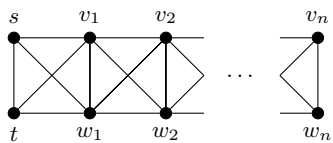
\begin{figure}
    \centering
    \begin{tikzpicture}[scale=1]
        \footnotesize
        %\draw[rounded corners = 5mm] (-0.75,-1.5) rectangle (11.75,2.5);
        \node[vertex, label=90:$s$] (s) at (0,1) {};
        \node[vertex, label=-90:$t$] (s') at (0,0) {};
        %\node[vertex, label=0:$z$] (z) at (5.5,3.25) {};
        \node[vertex, label=90:$v_1$] (x1) at (1,1) {};
        \node[vertex, label=-90:$w_1$] (nx1) at (1,0) {};
        \node[vertex, label=90:$v_2$] (x2) at (2,1) {};
        \node[vertex, label=-90:$w_2$] (nx2) at (2,0) {};
        \node[vertex, label=90:$v_n$] (xn) at (4,1) {};
        \node[vertex, label=-90:$w_n$] (nxn) at (4,0) {};

        \draw (s) -- (x1) -- (nx1) -- (s);
        \draw (s') -- (x1) -- (nx1) -- (s') -- (s);
        \draw (x1) -- (x2) --
               (nx2) -- (nx1) -- (x2) -- (nx1) -- (x1) -- (nx2) -- (x2);
        \draw (x2) --+ (0.5, 0);
        \draw (nx2) --+ (0.5, 0);
        \draw (x2) --+ (0.5, -0.5);
        \draw (nx2) --+ (0.5, 0.5);
        
        \node at (3,0.5) {$\dots$};

        \draw (xn) -- (nxn);
        \draw (xn) --+ (-0.5,0);
        \draw (nxn) --+ (-0.5, 0);
        \draw (xn) --+ (-0.5, -0.5);
        \draw (nxn) --+ (-0.5, 0.5);
\end{tikzpicture}
    \caption{A graph of \param{bandwidth}~$3$ whose Hamiltonian paths have an exponential number of path mappings to the path decomposition.}
    \label{fig:path-example}
\end{figure}

\subsection{Treewidth}\label{subsec:treewidth2}
In this section we consider graphs that have \param{treewidth} at most \(3\).
We reuse some ideas from \cref{subsec:pathwidth3} together with some additional observations.
Let \(\mathcal{T}=(T,\{X_t\}_{t\in V(T)})\) be a tree decomposition of width~\(3\) for \(G\). 
Without loss of generality, we assume that every bag contains exactly \(4\) vertices and that \(\mathcal{T}\) is a rooted binary tree. 
The nodes with two children in \(\mathcal{T}\) are called \emph{join nodes} and the other nodes are called \emph{exchange nodes}.
Let \(t\) be a join node with children \(t_1, t_2\). Then \(X_t = X_{t_2} = X_{t_2}\).
We also assume that no leaf of \(\mathcal{T}\) is a child of a join node.
For the exchange nodes, we make the same assumption as for the vertex sets in \cref{subsec:pathwidth3}, namely that they differ in exactly two vertices. Let \(t_1\) be an exchange node with child \(t_2\), then the unique vertex in \(X_{t_1}\setminus X_{t_2}\) is introduced in \(X_{t_1}\) and the vertex in \(X_{t_2}\setminus X_{t_1}\) is forgotten in \(X_{t_1}\).
We can compute a nice tree decomposition of width at most \(3\) in linear time \cite{bodlaender1996linear-time}. Given a nice tree decomposition of width at most \(3\), a tree decomposition of the form above can be found in linear time.

For a node \(t\), let \(V_t\) be the union of \(X_{t'}\) for all \(t'\) in the subtree of \(\mathcal{T}\) rooted at \(t\).
Furthermore, let \(G_t\) be the induced subgraph on \(V_t\).
Let \(\C\) be an ordered Hamiltonian cycle in \(G\) and let \(\C_t= G_t \cap \C\). We reuse the definitions of start vertex, end vertex, interior vertex, terminal vertex,  trivial paths as well as midpart and close-part given in the section before.

\begin{lemma}\label{lem:tw2_form_non-join}
    Let \(\C\) be a Hamiltonian cycle in \(G\) and \(t\) be a node whose parent is an exchange node. Then \(\C_t\) consists of one non-trivial path together with at most one trivial path.
\end{lemma}
\begin{proof}
We can use \cref{lem:pw3_storinternal} since \(t\) behaves similar as a node in a path decomposition. Thus, the vertex $u$ that is forgotten in the parent of $t$ is an interior vertex in $t$. As the bag contains four vertices, there is at most one vertex that is not part of the path of $u$ and this vertex has to form a trivial path.
\end{proof}

\begin{lemma}\label{lem:tw2_form_join}
Let \(\C\) be a Hamiltonian cycle and let \(t\) be a node whose parent is a join node. Then \(\C_t\) either consists of one non-trivial path together with at most one trivial path or it contains two non-trivial paths.
\end{lemma}
\begin{proof}
The statement follows directly from the fact that there are at most four terminal vertices in the bag.
\end{proof}

\begin{lemma}\label{lem:tw2_mid+other}
    Let \(\C\) be an ordered Hamiltonian cycle and let \(t\) be a join node, such that \(\C_t\) contains two non-trivial paths.
    Let \(t'\) be the parent of \(t\) and \(t_1,t_2\) be the children of \(t\).
    Then \(\C_{t'}\) contains exactly one path.
    Furthermore, \(\C_{t_1}\) and \(\C_{t_2}\) contain exactly one non-trivial path each.
\end{lemma}
\begin{proof}
    By the form of the tree decomposition, \(t_1, t_2\) are not leaves of \(\mathcal{T}\) and, thus, there is at least one vertex in \(V_{t_i}\setminus X_{t_i}\). 
    This implies that \(\C_{t_i}\) contains at least one non-trivial path for each \(i\in \{1,2\}\).

Now assume that for one of the children, say \(t_1\), \(\C_{t_1}\) contains two non-trivial paths \(P_1\) and \(P_2\) and \(\C_{t_2}\) contains a path \(P_3\). 
As there are only four vertices in $X_t$, one of the paths must join the other two paths in $X_t$, a contradiction to the form of~\(\C_t\). Hence $C_{t_1}$ and $\C_{t_2}$ contain exactly one non-trivial path each.

    Now consider the parent \(t'\) of \(t\). Due to \cref{lem:tw2_form_non-join}, $t'$ has to be a join node. Let $t''$ be the sibling of \(t\). Again, $\C_{t''}$ is not a leaf and, thus, contains at least one non-trivial path. As above, the two paths of $\C_t$ and the one path of $C_{t''}$ have to join to one non-trivial path that contains all four vertices of $X_{t'}$.
 \end{proof}

As in \cref{subsec:pathwidth3}, we define the signature \(\sigma\) of a bag \(X_t\) as a mapping of the vertices to the possible types of vertices and paths.
Additionally, if the signature induces two non-trivial paths, \(\sigma\) stores an additional bit \(d=\{L,R\}\).
This bit is used to encode if the path with number~1 came from the left or from the right subtree.
A signature is \emph{valid} if it has one of the forms stated by \cref{lem:tw2_form_non-join} and \cref{lem:tw2_form_join}.
Let \(\sigma\) be a valid signature for a node \(t\).
If \(t\) is an exchange node, we call a signature \(\gamma\) for its child \emph{compatible} if \(\gamma\) is valid and there is a way to extend the paths induced by \(\gamma\) with the vertex introduced in \(X_t\).
Similarly, if \(t\) is a join node, we call the signatures \(\gamma_1, \gamma_2\) for its left and right child compatible to \(\sigma\) if joining the paths induced by \(\gamma_1, \gamma_2\) gives the paths induced by \(\sigma\).
In particular, if the path with number~1 in $\sigma$ is a close-part and \(d=L\), then \(\gamma_1\) has to induce a close-part. In the other case (\(d=R\)), \(\gamma_2\) has to induce a close-part.

Given a signature \(\sigma\) for \(X_t\), a \emph{path mapping} \(f_\sigma\) assigns each vertex \(v\in V_t\) to one of the paths induced by \(\sigma\).
Again, a path mapping \emph{contradicts} \(\pi\) if there is no way to append or prepend the vertices in \(V\setminus V_t\) to the paths without violating the constraints given by \(\pi\).

\begin{lemma}\label{lem:tw2_unique_compatible}
Let \(\sigma\) be a signature for \(X_t\). There is at most one possible path mapping \(f_\sigma\).
\end{lemma}
\begin{proof}
If $\sigma$ is a leaf or induces only one non-trivial path, then the statement is clear. If it induces two non-trivial paths, then $t$ is a join node, due to \cref{lem:tw2_form_non-join}. Let $t_1$ and $t_2$ be the children of $t$ and let $\gamma_1$ and $\gamma_2$ be signatures that are compatible with $\sigma$. As there are no interior vertices in $\sigma$, the terminal vertices of a non-trivial paths of $\sigma$ are terminal vertices of the same path in one of $\gamma_1$ and $\gamma_2$ and trivial paths in the other. The value \(d(\sigma)\) directly implies which non-trivial path has to be in the left child and which in the right child. As the path mapping of $\gamma_1$ and $\gamma_2$ are unique, the same holds for the path mapping of $\sigma$.
\end{proof}

We now have all the definitions to define the dynamic program for POHCP. 
For each bag \(X_t\) and each valid signature \(\sigma\) for \(X_t\), define the subproblem \(D[t,\sigma]\).
If there is a path mapping \(f_\sigma\) for a signature \(\sigma\) of \(X_t\), then \(D[t,\sigma]=f_\sigma\), otherwise \(D[t,\sigma] = \bot\).
The table is filled bottom up along the tree. 
For a leaf there is only one possible path mapping where it has to be checked if it contradicts \(\pi\).

For an exchange node, the same algorithm as in \cref{subsec:pathwidth3} can be used. 
Now consider a join node \(t\) with signature \(\sigma\).
Iterate over all compatible signatures \(\gamma_1,\gamma_2\) for its children and consider the uniquely defined path partition induced by this triple. 
If there is one pair of compatible signatures that induce a path partition \(f_\sigma\) that does not contradict \(\pi\), set \(D[t,\sigma]\) to this value, otherwise set \(D[t,\sigma]= \bot\).

\begin{theorem}\label{thm:tw3}
    MinPOHCP in graphs of \param{treewidth} at most \(3\) can be solved in \(\O(n^3)\) time.
\end{theorem}
\begin{proof}
There are \(\O(n)\) bags as every vertex is forgotten at most once. 
Furthermore, for each bag, there are \(\O(1)\) possible valid signatures. 
Consequentially, there can be only \(\O(1)\) compatible signatures to check for each node. Similar as in the proof of \cref{thm:pw4}, this check might take $\O(n^2)$ time. Consequently, we need $\O(n^3)$ time in total.

For correctness, observe that by \cref{lem:tw2_mid+other}, if there is a signature \(\sigma\) for a node \(t\) that induces two non-trivial paths, 
\cref{lem:tw2_form_non-join} implies that this node is a join node and by \cref{lem:tw2_unique_compatible} there is only one pair \(\gamma_1,\gamma_2\) of compatible signatures for \(\sigma\).
Furthermore, by \cref{lem:tw2_mid+other}, the parent of \(t\) induces exactly one non-trivial path.
Thus, the information about which side the two paths came from is not needed further above in the tree and can safely be ignored.

By storing the path partition with minimum weight if there is more than one candidate, the algorithm above extends to MinPOHCP.
\end{proof}

Similar as for \param{pathwidth}, we can use \cref{obs:path-to-cycle} to derive an algorithm for the path variant on graphs of \param{treewidth} at most $2$.

\begin{theorem}\label{thm:tw2}
    MinPOHPP on graphs of \param{treewidth} at most \(2\) can be solved in \(\O(n^3)\) time.  
\end{theorem}

Similar as for \cref{corol:number-path-pw}, the statement of \cref{lem:tw2_unique_compatible} implies a linear bound on the number of path mappings of Hamiltonian cycles.

\begin{corollary}
    Let $G$ be an $n$-vertex graph of \param{treewidth} at most $3$ and let $(T,(X_t)_{t \in V(T)})$ be a tree decomposition of $G$ of width at most $3$. Let $\mathfrak{C}$ be the set of all Hamiltonian cycles in $G$. The size of the family of path mappings $\{\C \cap V_i \mid \C \in \mathfrak{C}, i \in \{1,\dots,k\}, V_i = \bigcup_{j=1}^i X_j\}$ is $\O(n)$. 
\end{corollary}

\subparagraph{A note on treewidth 4} 
When seeing \cref{thm:pw4} and \cref{thm:tw3}, one might hope to combine and extend the algorithms to give an algorithm for POHCP on graphs of \param{treewidth} at most~\(4\) (or for POHPP on graphs of \param{treewidth} at most~$3$). 
This however seems to pose a bigger challenge as \cref{lem:tw2_unique_compatible} is not true anymore. 
Thus, the signature for nodes with two non-trivial paths has to maintain more than local information about the origin of each path.
In the algorithm of \cref{thm:pw4}, this global information was efficiently encoded by relying on the linear structure of the path decomposition.
This seems not to be that easy for graphs of \param{treewidth} \(4\), leaving an interesting avenue for further research.

\section{Conclusion}

We have considered the complexity of finding Hamiltonian paths and cycles with precedence constraints for graphs of bounded \param{bandwidth}, \param{pathwidth}, and \param{treewidth}. We have completely settled the complexity status of both problems for \param{bandwidth} and \param{pathwidth}. However, as can been seen in \cref{fig:results}, for \param{treewidth} there are some open cases. Therefore, one of our main open questions is the complexity of (Min)POHPP on graphs of \param{treewidth}~3 as well as of (Min)POHCP on graphs of \param{treewidth}~4.

Besides the general cases, we have also considered the restricted case of rectangular grid graphs of bounded height. We could show that all problems are hard for some fixed height value. As shown in \cref{fig:results}, the complexity results for grid graphs are more diffuse and the gaps between the polynomial-time results and the hardness results are larger. In particular, there are height values for which a weighted problem is hard while the complexity of the unweighted problem is open. These cases are of particular interest for the following reason: It is easy to see that MinPOHPP and MinPOHCP are hard on complete graphs while their unweighted variants are trivial. To the best of our knowledge, there is no \emph{sparse} graph class known for which the same holds. As grid graphs are sparse, rectangular grid graphs of height~5 (for paths) and~6 (for cycles) might be an example of such a class.

If the complexity gaps of rectangular grid graphs and graphs of bounded pathwidth turn out to be different, it would also be interesting to study the \emph{planar} graphs of bounded pathwidth to find the point where the complexity changes.  

Further research could also explore faster (linear-time) algorithms for grid graphs of height 3, graphs of pathwidth~2 or 3, graphs of treewidth~2 and proper interval graphs with clique number less than 5.

\bibliographystyle{plainurl}
\bibliography{lit}

\begin{thebibliography}{10}

\bibitem{ahmed2001travelling}
Zakir~Hussain Ahmed and S.~N.~Narahari Pandit.
\newblock The travelling salesman problem with precedence constraints.
\newblock {\em Opsearch}, 38(3):299--318, 2001.
\newblock \href {https://doi.org/10.1007/BF03398638}
  {\path{doi:10.1007/BF03398638}}.

\bibitem{ascheuer1993cutting}
Norbert Ascheuer, Laureano~F. Escudero, Martin Gr\"{o}tschel, and Mechthild
  Stoer.
\newblock A cutting plane approach to the sequential ordering problem (with
  applications to job scheduling in manufacturing).
\newblock {\em SIAM Journal on Optimization}, 3(1):25--42, 1993.
\newblock \href {https://doi.org/10.1137/0803002} {\path{doi:10.1137/0803002}}.

\bibitem{asdre2010fixed}
Katerina Asdre and Stavros~D. Nikolopoulos.
\newblock The 1-fixed-endpoint path cover problem is polynomial on interval
  graphs.
\newblock {\em Algorithmica}, 58(3):679--710, 2010.
\newblock \href {https://doi.org/10.1007/s00453-009-9292-5}
  {\path{doi:10.1007/s00453-009-9292-5}}.

\bibitem{asdre2010polynomial}
Katerina Asdre and Stavros~D. Nikolopoulos.
\newblock A polynomial solution to the $k$-fixed-endpoint path cover problem on
  proper interval graphs.
\newblock {\em Theoretical Computer Science}, 411(7):967--975, 2010.
\newblock \href {https://doi.org/10.1016/j.tcs.2009.11.003}
  {\path{doi:10.1016/j.tcs.2009.11.003}}.

\bibitem{beisegel2024computing}
Jesse Beisegel, Fabienne Ratajczak, and Robert Scheffler.
\newblock Computing {H}amiltonian paths with partial order restrictions.
\newblock {\em ACM Transactions on Computation Theory}, 17(1), 2025.
\newblock \href {https://doi.org/10.1145/3711844} {\path{doi:10.1145/3711844}}.

\bibitem{bertossi1983finding}
Alan~A. Bertossi.
\newblock Finding {H}amiltonian circuits in proper interval graphs.
\newblock {\em Information Processing Letters}, 17(2):97--101, 1983.
\newblock \href {https://doi.org/10.1016/0020-0190(83)90078-9}
  {\path{doi:10.1016/0020-0190(83)90078-9}}.

\bibitem{bianco1994exact}
Lucio Bianco, Aristide Mingozzi, Salvatore Ricciardelli, and Massimo Spadoni.
\newblock Exact and heuristic procedures for the traveling salesman problem
  with precedence constraints, based on dynamic programming.
\newblock {\em INFOR: Information Systems and Operational Research},
  32(1):19--32, 1994.
\newblock \href {https://doi.org/10.1080/03155986.1994.11732235}
  {\path{doi:10.1080/03155986.1994.11732235}}.

\bibitem{bodlaender1996linear-time}
Hans~L. Bodlaender.
\newblock A linear-time algorithm for finding tree-decompositions of small
  treewidth.
\newblock {\em SIAM Journal on Computing}, 25(6):1305--1317, 1996.
\newblock \href {https://doi.org/10.1137/S0097539793251219}
  {\path{doi:10.1137/S0097539793251219}}.

\bibitem{bodlaender2015deterministic}
Hans~L. Bodlaender, Marek Cygan, Stefan Kratsch, and Jesper Nederlof.
\newblock Deterministic single exponential time algorithms for connectivity
  problems parameterized by treewidth.
\newblock {\em Information and Computation}, 243:86--111, 2015.
\newblock \href {https://doi.org/10.1016/J.IC.2014.12.008}
  {\path{doi:10.1016/J.IC.2014.12.008}}.

\bibitem{chen2002efficient}
Shao~Dong Chen, Hong Shen, and Rodney~W. Topor.
\newblock An efficient algorithm for constructing {H}amiltonian paths in
  meshes.
\newblock {\em Parallel Comput.}, 28(9):1293--1305, 2002.
\newblock \href {https://doi.org/10.1016/S0167-8191(02)00135-7}
  {\path{doi:10.1016/S0167-8191(02)00135-7}}.

\bibitem{colbourn1985minimizing}
Charles~J. Colbourn and William~R. Pulleyblank.
\newblock Minimizing setups in ordered sets of fixed width.
\newblock {\em Order}, 1:225--229, 1985.
\newblock \href {https://doi.org/10.1007/BF00383598}
  {\path{doi:10.1007/BF00383598}}.

\bibitem{courcelle1990monadic}
Bruno Courcelle.
\newblock The monadic second-order logic of graphs. {I}. {R}ecognizable sets of
  finite graphs.
\newblock {\em Information and Computation}, 85(1):12--75, 1990.
\newblock \href {https://doi.org/10.1016/0890-5401(90)90043-H}
  {\path{doi:10.1016/0890-5401(90)90043-H}}.

\bibitem{courcelle1992monadic}
Bruno Courcelle.
\newblock The monadic second-order logic of graphs {III}:
  {T}ree-decompositions, minors and complexity issues.
\newblock {\em Informatique Th\'eorique et Applications/Theoretical Informatics
  and Applications}, 26(3):257--286, 1992.
\newblock URL: \url{http://www.numdam.org/item/ITA_1992__26_3_257_0/}.

\bibitem{cygan2015param}
Marek Cygan, Fedor~V. Fomin, {\L}ukasz Kowalik, Daniel Lokshtanov, D{\'a}niel
  Marx, Marcin Pilipczuk, Micha{\l} Pilipczuk, and Saket Saurabh.
\newblock {\em Parameterized Algorithms}.
\newblock Springer, Cham, 2015.
\newblock \href {https://doi.org/10.1007/978-3-319-21275-3}
  {\path{doi:10.1007/978-3-319-21275-3}}.

\bibitem{cygan2018fast}
Marek Cygan, Stefan Kratsch, and Jesper Nederlof.
\newblock Fast {H}amiltonicity checking via bases of perfect matchings.
\newblock {\em Journal of the {ACM}}, 65(3):12:1--12:46, 2018.
\newblock \href {https://doi.org/10.1145/3148227} {\path{doi:10.1145/3148227}}.

\bibitem{cygan2022solving}
Marek Cygan, Jesper Nederlof, Marcin Pilipczuk, Micha{\l} Pilipczuk, Johan
  M.~M. van Rooij, and Jakub~Onufry Wojtaszczyk.
\newblock Solving connectivity problems parameterized by treewidth in single
  exponential time.
\newblock {\em {ACM} Transactions on Algorithms}, 18(2):17:1--17:31, 2022.
\newblock \href {https://doi.org/10.1145/3506707} {\path{doi:10.1145/3506707}}.

\bibitem{damschke1993paths}
Peter Damaschke.
\newblock Paths in interval graphs and circular arc graphs.
\newblock {\em Discrete Mathematics}, 112(1--3):49--64, 1993.
\newblock \href {https://doi.org/10.1016/0012-365X(93)90223-G}
  {\path{doi:10.1016/0012-365X(93)90223-G}}.

\bibitem{diestel}
Reinhard Diestel.
\newblock {\em Graph Theory}, volume 173 of {\em GTM}.
\newblock Springer, fourth edition, 2012.
\newblock \href {https://doi.org/10.1007/978-3-662-70107-2}
  {\path{doi:10.1007/978-3-662-70107-2}}.

\bibitem{escudero1988inexact}
Laureano~F. Escudero.
\newblock An inexact algorithm for the sequential ordering problem.
\newblock {\em European Journal of Operational Research}, 37(2):236--249, 1988.
\newblock \href {https://doi.org/10.1016/0377-2217(88)90333-5}
  {\path{doi:10.1016/0377-2217(88)90333-5}}.

\bibitem{escudero1988implementation}
Laureano~F. Escudero.
\newblock On the implementation of an algorithm for improving a solution to the
  sequential ordering problem.
\newblock {\em Trabajos de Investigacion-Operativa}, 3:117--140, 1988.

\bibitem{furer2016faster}
Martin F{\"{u}}rer.
\newblock Faster computation of path-width.
\newblock In Veli M{\"{a}}kinen, Simon~J. Puglisi, and Leena Salmela, editors,
  {\em Combinatorial Algorithms -- 27th International Workshop, {IWOCA} 2016},
  volume 9843 of {\em LNCS}, pages 385--396. Springer, 2016.
\newblock \href {https://doi.org/10.1007/978-3-319-44543-4_30}
  {\path{doi:10.1007/978-3-319-44543-4_30}}.

\bibitem{galil1978cyclic}
Zvi Galil and Nimrod Megiddo.
\newblock Cyclic ordering is {NP}-complete.
\newblock {\em Theoretical Computer Science}, 5(2):179--182, 1978.
\newblock \href {https://doi.org/10.1016/0304-3975(77)90005-6}
  {\path{doi:10.1016/0304-3975(77)90005-6}}.

\bibitem{lawler1985traveling}
Paul~C. Gilmore, Eugene~L. Lawler, and David~B. Shmoys.
\newblock Well-solved special cases.
\newblock In Eugene~L. Lawler, Jan~Karel Lenstra, Alexander H.~G. Rinnooy~Kan,
  and David~B. Shmoys, editors, {\em The traveling salesman problem. {A} guided
  tour of combinatorial optimization}, pages 87--143. Wiley, Chichester, 1985.

\bibitem{golovach2020graph}
Petr~A. Golovach, Ramaswamy Krithika, Abhishek Sahu, Saket Saurabh, and Meirav
  Zehavi.
\newblock Graph {H}amiltonicity parameterized by proper interval deletion set.
\newblock In Yoshiharu Kohayakawa and Fl{\'{a}}vio~Keidi Miyazawa, editors,
  {\em {LATIN} 2020: Theoretical Informatics -- 14th Latin American Symposium},
  volume 12118 of {\em LNCS}, pages 104--115. Springer, 2020.
\newblock \href {https://doi.org/10.1007/978-3-030-61792-9_9}
  {\path{doi:10.1007/978-3-030-61792-9_9}}.

\bibitem{hsieh2004efficient}
Sun-Yuan Hsieh.
\newblock An efficient parallel strategy for the two-fixed-endpoint
  {H}amiltonian path problem on distance-hereditary graphs.
\newblock {\em Journal of Parallel and Distributed Computing}, 64(5):662--685,
  2004.
\newblock \href {https://doi.org/10.1016/j.jpdc.2004.03.014}
  {\path{doi:10.1016/j.jpdc.2004.03.014}}.

\bibitem{itai1982hamilton}
Alon Itai, Christos~H. Papadimitriou, and Jayme~Luiz Szwarcfiter.
\newblock Hamilton paths in grid graphs.
\newblock {\em SIAM Journal on Computing}, 11(4):676--686, 1982.
\newblock \href {https://doi.org/10.1137/0211056} {\path{doi:10.1137/0211056}}.

\bibitem{kaplan1996pathwidth}
Haim Kaplan and Ron Shamir.
\newblock Pathwidth, bandwidth, and completion problems to proper interval
  graphs with small cliques.
\newblock {\em SIAM Journal on Computing}, 25(3):540--561, 1996.
\newblock \href {https://doi.org/10.1137/S0097539793258143}
  {\path{doi:10.1137/S0097539793258143}}.

\bibitem{keil1985finding}
J.~Mark Keil.
\newblock Finding {H}amiltonian circuits in interval graphs.
\newblock {\em Information Processing Letters}, 20(4):201--206, 1985.
\newblock \href {https://doi.org/10.1016/0020-0190(85)90050-X}
  {\path{doi:10.1016/0020-0190(85)90050-X}}.

\bibitem{keshavarz-kohjerdi2012hamiltonian}
Fatemeh Keshavarz{-}Kohjerdi and Alireza Bagheri.
\newblock Hamiltonian paths in some classes of grid graphs.
\newblock {\em Journal of Applied Mathematics}, 2012:475087, 2012.
\newblock \href {https://doi.org/10.1155/2012/475087}
  {\path{doi:10.1155/2012/475087}}.

\bibitem{keshavarz-kohjerdi2016hamiltonian}
Fatemeh Keshavarz-Kohjerdi and Alireza Bagheri.
\newblock Hamiltonian paths in {$L$}-shaped grid graphs.
\newblock {\em Theoretical Computer Science}, 621:37--56, 2016.
\newblock \href {https://doi.org/10.1016/j.tcs.2016.01.024}
  {\path{doi:10.1016/j.tcs.2016.01.024}}.

\bibitem{keshavarz-kohjerdi2020linear}
Fatemeh Keshavarz{-}Kohjerdi and Alireza Bagheri.
\newblock Linear-time algorithms for finding {H}amiltonian and longest $(s,
  t)$-paths in {$C$}-shaped grid graphs.
\newblock {\em Discrete Optimization}, 35:100554, 2020.
\newblock \href {https://doi.org/10.1016/J.DISOPT.2019.100554}
  {\path{doi:10.1016/J.DISOPT.2019.100554}}.

\bibitem{korte2018combinatorial}
Bernhard Korte and Jens Vygen.
\newblock {\em Combinatorial Optimization: Theory and Algorithms}.
\newblock Springer, Berlin, Heidelberg, 6th edition, 2018.
\newblock \href {https://doi.org/10.1007/978-3-662-56039-6}
  {\path{doi:10.1007/978-3-662-56039-6}}.

\bibitem{li2017linear}
Peng Li and Yaokun Wu.
\newblock A linear time algorithm for the 1-fixed-endpoint path cover problem
  on interval graphs.
\newblock {\em SIAM Journal on Discrete Mathematics}, 31(1):210--239, 2017.
\newblock \href {https://doi.org/10.1137/140981265}
  {\path{doi:10.1137/140981265}}.

\bibitem{lutz1996cook}
Jack~H. Lutz and Elvira Mayordomo.
\newblock Cook versus {K}arp-{L}evin: {S}eparating completeness notions if {NP}
  is not small.
\newblock {\em Theoretical Computer Science}, 164(1):141--163, 1996.
\newblock \href {https://doi.org/10.1016/0304-3975(95)00189-1}
  {\path{doi:10.1016/0304-3975(95)00189-1}}.

\bibitem{mandal2014separating}
Debasis Mandal, Aduri Pavan, and Rajeswari Venugopalan.
\newblock Separating {C}ook completeness from {K}arp-{L}evin completeness under
  a worst-case hardness hypothesis.
\newblock In Venkatesh Raman and S.~P. Suresh, editors, {\em 34th International
  Conference on Foundation of Software Technology and Theoretical Computer
  Science, FSTTCS 2014}, volume~29 of {\em LIPIcs}, pages 445--456. Schloss
  Dagstuhl -- Leibniz-Zentrum f{\"u}r Informatik, 2014.
\newblock \href {https://doi.org/10.4230/LIPIcs.FSTTCS.2014.445}
  {\path{doi:10.4230/LIPIcs.FSTTCS.2014.445}}.

\bibitem{mertzios2010optimal}
George~B. Mertzios and Walter Unger.
\newblock An optimal algorithm for the $k$-fixed-endpoint path cover on proper
  interval graphs.
\newblock {\em Mathematics in Computer Science}, 3(1):85--96, 2010.
\newblock \href {https://doi.org/10.1007/s11786-009-0004-y}
  {\path{doi:10.1007/s11786-009-0004-y}}.

\bibitem{monien1980bounding}
Burkhard Monien and Ivan~Hal Sudborough.
\newblock Bounding the bandwidth of {NP}-complete problems.
\newblock In Hartmut Noltemeier, editor, {\em Graphtheoretic Concepts in
  Computer Science -- International Workshop, {WG} 80}, volume 100 of {\em
  LNCS}, pages 279--292. Springer, 1980.
\newblock \href {https://doi.org/10.1007/3-540-10291-4_20}
  {\path{doi:10.1007/3-540-10291-4_20}}.

\bibitem{porschen2007algorithms}
Stefan Porschen and Ewald Speckenmeyer.
\newblock Algorithms for variable-weighted 2-{SAT} and dual problems.
\newblock In Jo{\~{a}}o Marques{-}Silva and Karem~A. Sakallah, editors, {\em
  Theory and Applications of Satisfiability Testing -- {SAT} 2007}, volume 4501
  of {\em LNCS}, pages 173--186. Springer, 2007.
\newblock \href {https://doi.org/10.1007/978-3-540-72788-0_19}
  {\path{doi:10.1007/978-3-540-72788-0_19}}.

\bibitem{scheffler2025semi}
Robert Scheffler.
\newblock Semi-proper interval graphs.
\newblock {\em Discrete Applied Mathematics}, 360:22--41, 2025.
\newblock \href {https://doi.org/10.1016/J.DAM.2024.08.016}
  {\path{doi:10.1016/J.DAM.2024.08.016}}.

\bibitem{syslo1979characterizations}
Maciej~M. Sysło.
\newblock Characterizations of outerplanar graphs.
\newblock {\em Discrete Mathematics}, 26(1):47--53, 1979.
\newblock \href {https://doi.org/10.1016/0012-365X(79)90060-8}
  {\path{doi:10.1016/0012-365X(79)90060-8}}.

\bibitem{umans1997hamiltonian}
Christopher Umans and William~J. Lenhart.
\newblock Hamiltonian cycles in solid grid graphs.
\newblock In {\em 38th Annual Symposium on Foundations of Computer Science,
  {FOCS} '97}, pages 496--505. {IEEE} Computer Society, 1997.
\newblock \href {https://doi.org/10.1109/SFCS.1997.646138}
  {\path{doi:10.1109/SFCS.1997.646138}}.

\bibitem{yeh1998path}
Hong-Gwa Yeh and Gerard~J. Chang.
\newblock The path-partition problem in bipartite distance-hereditary graphs.
\newblock {\em Taiwanese Journal of Mathematics}, 2(3):353--360, 1998.
\newblock \href {https://doi.org/10.11650/twjm/1500406975}
  {\path{doi:10.11650/twjm/1500406975}}.

\bibitem{ziobro2019finding}
Micha{\l} Ziobro and Marcin Pilipczuk.
\newblock Finding {H}amiltonian cycle in graphs of bounded treewidth:
  {E}xperimental evaluation.
\newblock {\em {ACM} Journal of Experimental Algorithmics}, 24:2.7:1--2.7:18,
  2019.
\newblock \href {https://doi.org/10.1145/3368631} {\path{doi:10.1145/3368631}}.

\end{thebibliography}

\end{document}